\documentclass[pra,aps,nopacs,onecolumn,twoside,superscriptaddress]{revtex4}

\usepackage{amsmath,amsfonts,amssymb,caption,color,epsfig,graphics,graphicx,hyperref,latexsym,mathrsfs,revsymb,theorem,url,verbatim,epstopdf,mathtools,enumerate}
\usepackage{subfigure}\usepackage{makecell}\usepackage{multirow}\usepackage{diagbox}
\hypersetup{colorlinks,linkcolor={blue},citecolor={blue},urlcolor={red}}

\usepackage[qm]{qcircuit}

\newtheorem{definition}{Definition}
\newtheorem{proposition}[definition]{Proposition}
\newtheorem{lemma}[definition]{Lemma}

\newtheorem{theorem}[definition]{Theorem}
\newtheorem{corollary}[definition]{Corollary}
\newtheorem{conjecture}[definition]{Conjecture}

\newtheorem{remark}[definition]{Remark}
\newtheorem{example}[definition]{Example}
\newtheorem{question}[definition]{Question}
\newtheorem{memo}[definition]{Memo}

\def\squareforqed{\hbox{\rlap{$\sqcap$}$\sqcup$}}
\def\qed{\ifmmode\squareforqed\else{\unskip\nobreak\hfil
\penalty50\hskip1em\null\nobreak\hfil\squareforqed
\parfillskip=0pt\finalhyphendemerits=0\endgraf}\fi}
\def\endenv{\ifmmode\;\else{\unskip\nobreak\hfil
\penalty50\hskip1em\null\nobreak\hfil\;
\parfillskip=0pt\finalhyphendemerits=0\endgraf}\fi}
\newenvironment{proof}{\noindent \textbf{{Proof.~} }}{\qed}
\def\Dbar{\leavevmode\lower.6ex\hbox to 0pt
{\hskip-.23ex\accent"16\hss}D}
\makeatletter
\def\url@leostyle{%
  \@ifundefined{selectfont}{\def\UrlFont{\sf}}{\def\UrlFont{\small\ttfamily}}}
\makeatother
\def\bcj{\begin{conjecture}}
\def\ecj{\end{conjecture}}
\def\bcr{\begin{corollary}}
\def\ecr{\end{corollary}}
\def\bd{\begin{definition}}
\def\ed{\end{definition}}
\def\bea{\begin{eqnarray}}
\def\eea{\end{eqnarray}}
\def\beq{\begin{equation}}
\def\eeq{\end{equation}}
\def\bal{\begin{aligned}}
\def\eal{\end{aligned}}
\def\bem{\begin{enumerate}}
\def\eem{\end{enumerate}}
\def\bex{\begin{example}}
\def\eex{\end{example}}
\def\bim{\begin{itemize}}
\def\eim{\end{itemize}}
\def\bl{\begin{lemma}}
\def\el{\end{lemma}}
\def\bma{\begin{bmatrix}}
\def\ema{\end{bmatrix}}
\def\bpf{\begin{proof}}
\def\epf{\end{proof}}
\def\bpp{\begin{proposition}}
\def\epp{\end{proposition}}
\def\bqu{\begin{question}}
\def\equ{\end{question}}
\def\br{\begin{remark}}
\def\er{\end{remark}}
\def\bt{\begin{theorem}}
\def\et{\end{theorem}}
\def\bmm{\begin{memo}}
\def\emm{\end{memo}}

\def\btb{\begin{tabular}}
\def\etb{\end{tabular}}

\newcommand{\nc}{\newcommand}

\def\a{\alpha}
\def\b{\beta}
\def\g{\gamma}

\def\t{\theta}

\def\l{\lambda}

\def\r{\rho}
\def\s{\sigma}

\def\og{\omega}

\def\L{\Lambda}

\nc{\bbA}{\mathbb{A}} \nc{\bbB}{\mathbb{B}} \nc{\bbC}{\mathbb{C}}
 \nc{\bbD}{\mathbb{D}} \nc{\bbE}{\mathbb{E}} \nc{\bbF}{\mathbb{F}}
 \nc{\bbG}{\mathbb{G}} \nc{\bbH}{\mathbb{H}} \nc{\bbI}{\mathbb{I}}
 \nc{\bbJ}{\mathbb{J}} \nc{\bbK}{\mathbb{K}} \nc{\bbL}{\mathbb{L}}
 \nc{\bbM}{\mathbb{M}} \nc{\bbN}{\mathbb{N}} \nc{\bbO}{\mathbb{O}}
 \nc{\bbP}{\mathbb{P}} \nc{\bbQ}{\mathbb{Q}} \nc{\bbR}{\mathbb{R}}
 \nc{\bbS}{\mathbb{S}} \nc{\bbT}{\mathbb{T}} \nc{\bbU}{\mathbb{U}}
 \nc{\bbV}{\mathbb{V}} \nc{\bbW}{\mathbb{W}} \nc{\bbX}{\mathbb{X}}
 \nc{\bbZ}{\mathbb{Z}}

 \nc{\bA}{{\bf A}} \nc{\bB}{{\bf B}} \nc{\bC}{{\bf C}}
 \nc{\bD}{{\bf D}} \nc{\bE}{{\bf E}} \nc{\bF}{{\bf F}}
 \nc{\bG}{{\bf G}} \nc{\bH}{{\bf H}} \nc{\bI}{{\bf I}}
 \nc{\bJ}{{\bf J}} \nc{\bK}{{\bf K}} \nc{\bL}{{\bf L}}
 \nc{\bM}{{\bf M}} \nc{\bN}{{\bf N}} \nc{\bO}{{\bf O}}
 \nc{\bP}{{\bf P}} \nc{\bQ}{{\bf Q}} \nc{\bR}{{\bf R}}
 \nc{\bS}{{\bf S}} \nc{\bT}{{\bf T}} \nc{\bU}{{\bf U}}
 \nc{\bV}{{\bf V}} \nc{\bW}{{\bf W}} \nc{\bX}{{\bf X}}
 \nc{\bZ}{{\bf Z}}

\nc{\cA}{{\cal A}} \nc{\cB}{{\cal B}} \nc{\cC}{{\cal C}}
\nc{\cD}{{\cal D}} \nc{\cE}{{\cal E}} \nc{\cF}{{\cal F}}
\nc{\cG}{{\cal G}} \nc{\cH}{{\cal H}} \nc{\cI}{{\cal I}}
\nc{\cJ}{{\cal J}} \nc{\cK}{{\cal K}} \nc{\cL}{{\cal L}}
\nc{\cM}{{\cal M}} \nc{\cN}{{\cal N}} \nc{\cO}{{\cal O}}
\nc{\cP}{{\cal P}} \nc{\cQ}{{\cal Q}} \nc{\cR}{{\cal R}}
\nc{\cS}{{\cal S}} \nc{\cT}{{\cal T}} \nc{\cU}{{\cal U}}
\nc{\cV}{{\cal V}} \nc{\cW}{{\cal W}} \nc{\cX}{{\cal X}}
\nc{\cZ}{{\cal Z}}

\nc{\hA}{{\hat{A}}} \nc{\hB}{{\hat{B}}} \nc{\hC}{{\hat{C}}}
\nc{\hD}{{\hat{D}}} \nc{\hE}{{\hat{E}}} \nc{\hF}{{\hat{F}}}
\nc{\hG}{{\hat{G}}} \nc{\hH}{{\hat{H}}} \nc{\hI}{{\hat{I}}}
\nc{\hJ}{{\hat{J}}} \nc{\hK}{{\hat{K}}} \nc{\hL}{{\hat{L}}}
\nc{\hM}{{\hat{M}}} \nc{\hN}{{\hat{N}}} \nc{\hO}{{\hat{O}}}
\nc{\hP}{{\hat{P}}} \nc{\hR}{{\hat{R}}} \nc{\hS}{{\hat{S}}}
\nc{\hT}{{\hat{T}}} \nc{\hU}{{\hat{U}}} \nc{\hV}{{\hat{V}}}
\nc{\hW}{{\hat{W}}} \nc{\hX}{{\hat{X}}} \nc{\hZ}{{\hat{Z}}}

\nc{\hn}{{\hat{n}}}

\def\diag{\mathop{\rm diag}}
\def\dim{\mathop{\rm Dim}}

\def\max{\mathop{\rm max}}

\def\tr{\mathop{\rm Tr}}

\def\dg{\dagger}

\newcommand{\bra}[1]{\langle#1|}
\newcommand{\ket}[1]{|#1\rangle}

\newcommand{\braket}[2]{\langle#1|#2\rangle}

\def\Dbar{\leavevmode\lower.6ex\hbox to 0pt
{\hskip-.23ex\accent"16\hss}D}

\begin{document}

\large

\title{Controlled remote implementation of operations via graph states}

\author{Xinyu Qiu}\email[]{xinyuqiu@buaa.edu.cn}
\affiliation{LMIB(Beihang University), Ministry of education, and School of Mathematical Sciences, Beihang University, Beijing 100191, China}
\author{Lin Chen}\email[]{linchen@buaa.edu.cn (corresponding author)}
\affiliation{LMIB(Beihang University), Ministry of education, and School of Mathematical Sciences, Beihang University, Beijing 100191, China}
\affiliation{International Research Institute for Multidisciplinary Science, Beihang University, Beijing 100191, China}

\begin{abstract}
We propose protocols for controlled remote implementation of operations  with convincing control power. Sharing a $(2N+1)$-partite graph state, $2N$ participants collaborate to prepare the stator and realize the operation $\otimes_{j=1}^N\exp{[i\alpha_j\sigma_{n_{O_j}}]}$ on $N$ unknown states for distributed systems $O_j$, with the permission of a controller.
 All the implementation requirements of our protocol can be satisfied by means of local operations and classical communications, and the experimental feasibility is presented according to current techniques. We characterize the entanglement requirement of our protocol in terms of geometric measure of entanglement. It turns out to be economic to realize the control function from the perspective of entanglement cost.  Further we show that the control power of our protocol is reliable by positive operator valued measurement.
\end{abstract}
\maketitle

\section{Introduction}
As the unique resource, entanglement allows the emergence and development of quantum information processing, such as teleportation \cite{bennett1993teleport}, dense coding \cite{bennett1992communication} and  cryptography \cite{gisin2002quantum,portmann2022security}. Being expected to offer substantial speed-ups over classical counterparts, quantum computation has been paid a lot of attention \cite{deutsch1992rapid,shor1997polynomial}. Several challenges have surfaced in its actual construction, such as decoherence and dissipation, sufficiently manipulating a large number of qubits, and undesirable interactions \cite{alexeew2021quantum}.   To counter such challenges,  distributed quantum computation has been proposed \cite{cirac1999distributed,serafini2006distributed,cacciapuoti2020quantum}. It requires to transfer states from one place to the other and implement the operations on a remote state faithfully. The first requirement has been met by quantum teleportation \cite{bennett1993teleport,pan1997experimental, georgescu202225years, qiu2022quantum}. The second one has been tackled by remote implementation of operation (RIO) and controlled RIO (CRIO).  RIO means that the quantum operation performed on the sender's local system is able to act on an unknown state of a remote system that belongs to the receiver \cite{huelga2001quantum,huelga2002remote}. Then CRIO was proposed by extending RIO to multipartite case with controller \cite{wang2007combined}. The idea of that is to implement remote operations, but only with the permission of controller. It can definitely enhance the security of RIO.  Both of RIO and CRIO play an important role not only in distributed quantum computation, but also other tasks in remote quantum information processing such as  programming \cite{nielsen2000quantum}, operation sharing \cite{wang2013determistic} and remote state preparation \cite{bennett2001remote}. They are realized by local operation and classical communication (LOCC)  and consume  entanglement resource.  Some works concerning RIO have been presented and interesting progress has been made both theoretically \cite{reznik2002remote, chen2005probabilistic,wang2006remote,an2022joint} and experimentally \cite{guo2005teleporting, huelga2005remote, bhaskar2020experimental}. 
As for CRIO, it has been realized in terms of partially unknown quantum operations \cite{wang2007combined, he2014bidirectional}, various classes of bipartite unitary operations \cite{yu2016implementation}, arbitrary dimensional controlled phase gate \cite{gong2021control}, and operators on different remote photon
states \cite{an2022controlled}. Entangled states including Bell,  Greenberger-Horne-Zeilinger (GHZ), and five-qubit cluster states are employed as  channels in these protocols. Theoretically, many of these protocols can hardly  be scalable, or technically complicated. What's more, the highly entangled states they employed are susceptible to noise, so it is a great challenge to realize them  under experimental techniques.   

In this paper,  we propose CRIO protocols for the operations $\exp{[i\a\s_{n}]}$ with convincing control power. 
The diagram of our protocol via a graph state $\ket{h_3}$ is shown in FIG. \ref{figure:tripartite}. 
We generalize the protocol realizing RIO by stators \cite{reznik2002remote} into CRIO.  The stators in our protocol are constructed from shared graph states by LOCC, only with the permission of the controller. The protocol via a $(2N+1)$-partite graph state can be obtained by generalization, and that is used to implement remote operations on $N$ unknown states.  Graph states are employed as the channel in our protocol for three reasons. Firstly, they make it possible to realize control function and enable our protocol to be scalable. Secondly, they are a natural resource for much of quantum information, and  known to be most readily available multipartite resource in the laboratory \cite{bell2014experimental,yang2022sequential}. Thirdly, many of the graph states show advantages for being robust to noise \cite{beriegel2001persistent}. 
 All the implementation requirements of our protocol can be satisfied by means of LOCC. The experimental feasibility is presented in terms of current techniques. The local operations and measurements can be implemented by a diamond nanophotonic resonator containing SiV quantum memory  with an integrated microwave stripline \cite{bhaskar2020experimental}. 
 The entanglement requirement of our protocol is characterized in terms of geometric measure of entanglement. It is presented in Proposition \ref{pro:GM}. Compared with the former protocol in \cite{reznik2002remote}, ours is endowed with control function, while it only requires the same entanglement resource. Hence it is economic to realize the control function from the perspective of entanglement cost.  Besides, we analyze the control power in our protocol, and obtain that without the permission of the controller, other participants can hardly realize the remote implementation of operation. It is shown in Proposition \ref{pro:controlpower}.  Thus the control power of our protocol is convincing.  Our protocol shows advantage in stronger security, extensive applications, and advanced efficiency. It can contribute to improving the ability of distributed quantum computing and stimulate more research work on quantum information processing.

\begin{figure}[htp]
	\center{\includegraphics[width=8cm]  {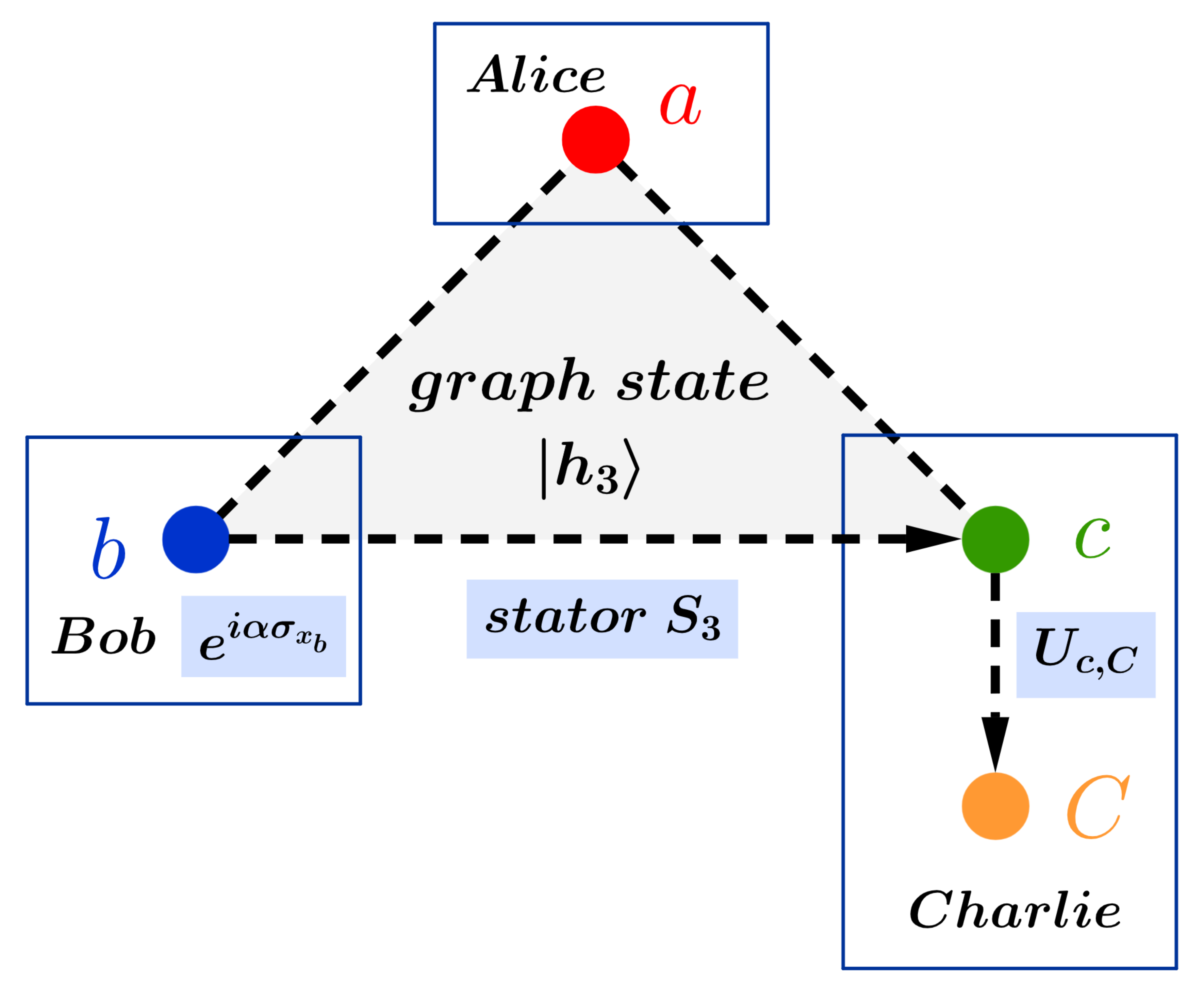}}
	\caption{Diagram showing CRIO via a graph state $\ket{h_3}_{a,b,c}$ in (\ref{def:stateh3}). Qubits $a,b,c$ belong to Alice, Bob and Charlie, respectively. These participants locate in distributed places and share the entangled state $\ket{h_3}$. Only with the permission of controller Alice, then Bob and Charlie can construct the stator by local measurements and operation $U_{c,C}$ etc. Then with the help of stator $S_3$, Bob can implement the remote operation $e^{i\a\s_{n_C}}$ on an unknown state for Charlie's system $C$ by locally implementing $e^{i\a\s_{x_b}}$.   }
	\label{figure:tripartite}
\end{figure}

Graph states have a strong connection with quantum computation. Naturally, they can be characterized by geometric measure of entanglement (GM), a well-known entanglement measure for multipartite systems.  GM not only provides a simple geometric picture, but also has significant operational meanings. It has connections with optimal entanglement witnesses \cite{wei2003geometric}, and multipartite state discrimination under LOCC \cite{hayashi2006bounds}. As one of the most widely used entanglement measures for the multipartite states, GM fulfills all the desired properties of an entanglement monotone \cite{wei2003geometric}. It has been utilized to determine the universality of resource states for one-way quantum computation \cite{nest2007fundamentals}. It also has been employed to show that most entangled states are too entangled  to be useful as computational resources \cite{gross2009most}.

The rest of this paper is organized as follows. In Sec. \ref{sec:preliminaries}, we introduce some basic concepts of graph states, GM and stator. Then we simply recall the deduction of eigenoperator equation of the stator. In Sec. \ref{sec:CRIO}, we propose our CRIO protocol.  We show the implementation of remote operation on an unknown state for a single system via a tripartite graph state in Sec. \ref{sec:tripartite}. The protocol via a five-partite graph state  is presented in Sec. \ref{sec:fivepartite}, which is slightly different from the former one  and used to implement operations on two remote systems. Then we generalize it into the one via a $(2N+1)$-partite graph state in Sec. \ref{sec:2n+1partite}. We show GM of the graph states used in our protocols in Sec. \ref{sec:calculateGM}. We do the control power analysis in Sec. \ref{sec:POVM}, and exhibit the experimental feasibility of our protocol in Sec. \ref{sec:experimental}. Finally, we conclude in Sec. \ref{sec:conclusion}.

\section{Preliminaries}
\label{sec:preliminaries}
In this section, we recall the definitions and some properties of graph states, GM and stator.  In Sec. \ref{sec:graphstate}, we recall the definition of graph states, and demonstrate a graph state with the help of quantum gates. In Sec. \ref{sec:geometricmeasure}, we recall the definition of GM and show a lemma we employ to characterize the GM of  graph states. In Sec. \ref{sec:stator}, we introduce the concept of stator, and present the eigenoperator equation of the stator used in this paper.
\subsection{Graph states}
\label{sec:graphstate}
A graph is a pair $G=(V,E)$, where $V$ is the set of vertices and $E\subset[V]^2$ is the set of edges. With each graph, a graph state is associated. An axiomatic framework for mapping graphs to quantum states is proposed in  \cite{ionicioiu2012encoding}. A graph state is a certain pure state on a Hilbert space $\mathcal{H}_2^{\otimes V}$. Each vertex of the graph labels a qubit.  Each vertex $a\in V$ of the graph $G=(V,E)$ is attached to a Hermitian operator
\begin{eqnarray}
K_G^{(a)}=\s_{x}^{(a)}\prod_{b\in N_a}\s_z^{(b)}.
\end{eqnarray} 
Here $\s_{x}^{(a)}$ and $\s_{z}^{(a)}$ are the Pauli matrices and the upper index specifies the Hilbert space on which the operator acts. $K_G^{(a)}$ is an observable of the qubits related to vertex $a$ and all of its neighbors $b\in N_a$. There are $N=|V|$ operators in the set $\{K_G^{(a)}\}_{a\in V}$. The operators in this set are all commute. 

We demonstrate a graph state with the help of Hadamard gate $H$ and two-qubit controlled-Z gate $CZ_{(j_1,j_2)}$, where
\begin{eqnarray}
	\label{def:Hgate}
	H=\frac{1}{\sqrt{2}}\bma1&1\\1&-1\ema,
\end{eqnarray}
and the gate
\begin{eqnarray}
	\label{def:CZgate}
CZ_{(j_1,j_2)}=\diag\{1,1,1,-1\}
\end{eqnarray}
denotes the gate with control qubit $j_1$ and controlled qubit $j_2$.
By the two gates, we can show the preparation of graph states in the quantum circuit conveniently.

A graph state $\ket{H_n}$  is created from a graph $G=\{V,E\}$ of $n$ vertices by assigning a qubit to each vertex and initializing them by applying the Hadamard gate on each qubit. Let $\ket{+}=(\ket{0}+\ket{1})/\sqrt{2}$. If two vertices $j_1,j_2\in V$ are connected by an edge $e\in V$, then we perform $CZ_{(j_1,j_2)}$ over the initialized $n$-qubit state $\ket{+}^{\otimes n}$. By implementing all the controlled-Z gates corresponding the edges $e\in E$, we obtain the graph state 
\begin{eqnarray}
\ket{H_n}=\prod_{e\in E}CZ_e\ket{+}^{\otimes n}.
\end{eqnarray}

Graph states are useful resources with applications spanning many aspects of quantum information processing, such as computation \cite{raussendorf2001a}, cryptography \cite{qian2012quantum}, quantum error correction \cite{liao2022topological} and  networks \cite{cuquet2012growth}. Experimentally, different techniques have been studied to implement graph states including ion traps \cite{barreiro2011An}, superconducting qubits \cite{song201710qubit}, and continuous variable optics \cite{walschaers2018tailoring}.

\subsection{Geometric measure of entanglement}
\label{sec:geometricmeasure}
 Geometric measure of entanglement (GM) is a well-known entanglement measure for multipartite systems \cite{wei2003geometric}. It measures the closest distance in terms of overlap between a given state and the set of separable states, or the set of pure product states. Originally introduced for pure bipartite states, GM was subsequently generalized to multipartite and to mixed states. Several  inequivalent definitions of GM has surfaced by now. In this paper, we shall follow the definition given in (\ref{def:GM}) and (\ref{def:GM2}).
\begin{eqnarray}
	\label{def:GM}
	\L^2(\r):=&&\max_{\s\in SEP}\tr(\r\s)=\max_{\ket{\varphi}\in PRO}\bra{\varphi}\r\ket{\varphi},\\
	\label{def:GM2}
	G(\r):=&&-2\log\L(\r).
\end{eqnarray}
Here SEP denotes the separable states and PRO denotes the fully pure product states in the Hilbert space $\otimes^{N}_{j=1}\mathcal{H}_j$. GM is known only for a few examples, such as bipartite pure states, GHZ-type states, antisymmetric basis states, pure symmetric three-qubit states and some graph states \cite{hayashi2008entanglement,markham2007entanglement,chen2010computation}.

  We show the following fact given in Ref. \cite{zhu2011additivity}. It is a useful lemma concerning the closest product states of non-negative states.  Here the \textit{closest product state}  denotes any pure product state maximizing (\ref{def:GM}). The \textit{non-negative state}  means that all its entries in the computational basis are non-negative. 
\begin{lemma}
	\label{lemma}
	The closest product state to a non-negative state $\r$ can be chosen to be non-negative.
\end{lemma}
The proof is presented in Lemma 8 of Ref. \cite{zhu2011additivity}. This lemma can be used to characterize GM of the states that are non-negative or locally equivalent to non-negative states. In addition, it contributes to prove the strong additivity of GM of the states including  Bell diagonal states, maximally correlated generalized Bell diagonal states, isotropic states, generalized Dicke states, mixture of Dicke states, the Smolin state and D\"ur's multipartite entangled states.

Since the graph states we employed in this paper are all locally equivalent to non-negative states, we use Lemma \ref{lemma} to investigate GM of them. The calculation is presented in Sec. \ref{sec:calculateGM}.

\subsection{Stator and eigenoperator equation}
\label{sec:stator}
The stator, a hybrid state operator,  is an object that expresses quantum correlations between states of one participant and operators of the other participant. It is firstly proposed in Ref. \cite{reznik2002remote} to implement a class of operations on remote systems. Given a well-prepared stator, the operation on Bob's system is remotely brought about by Alice's local operations. The desired operation is determined by Alice and unknown to Bob. Hence it demonstrates advantage in terms of security. 

A stator $S_{AB}$ shared by remote observers Alice and Bob is in the space
\begin{eqnarray}
S_{AB}\in  \{\mathcal{H}_A\otimes O(\mathcal{H}_B)\},
\end{eqnarray} 
where $\mathcal{H}_A$ and $\mathcal{H}_B$ are the Hilbert spaces of Alice and Bob respectively, and $O(\mathcal{H}_B)$ denotes the operators acting on an arbitrary state in $\mathcal{H_B}$.
A stator has the general form 
\begin{eqnarray}
S_{AB}=\sum_{s=1}^{N_A}\sum_{t=1}^{N_B^2}c_{st}\ket{s}\otimes B_t,
\end{eqnarray}
where $N_A=\dim(\mathcal{H}_A)$, $N_B=\dim(\mathcal{H}_B)$, $\ket{s}\in \mathcal{H}_A$, $B_t$ acts on states in $\mathcal{H}_B$ and $c_{s,t}$ are complex numbers.
For each stator, an eigenoperator equation can be constructed. We consider the   following stator  used in this paper,
\begin{eqnarray}
S=\ket{0}_A\otimes I_B+\ket{1}_A\otimes\s_{n_B}.
\end{eqnarray}
Here $\s_{n_B}=\overline{n_B}\cdot\overline{\s}\in O(\mathcal{H}_B)$, where $\overline{n_B}=[x,y,z]$ is the axis vector and $\overline{\s}=[\s_x,\s_y,\s_z]$ is the Pauli matrix vector. The states $\ket{0_A},\ket{1_A}\in \mathcal{H}_A$ are the eigenstates of $\s_{z_A}$.  One can verify that $\s^2_{n_B}=I$. Obviously, $S$ satisfies the eigenoperator equation 
\begin{eqnarray}
\s_{x_A}S=\s_{n_B}S.
\end{eqnarray}
Thus for any analytic function $f$, it also satisfies that
\begin{eqnarray}
f(\s_{x_A})S=f(\s_{n_B})S,
\end{eqnarray}
and particularly,
\begin{eqnarray}
	\label{eq:firststator}
e^{i\a\s_{x_A}}S=e^{i\a\s_{n_B}}S,
\end{eqnarray}
where $\a$ is any real number determined by Alice. Using the stator, a unitary operation on Alice's qubit gives rise to a similar unitary operation acting on Bob's system, which is remote to Alice. The construction of stator and the implementation of remote operations are both realized by LOCC.

\section{Controlled remote implementation of operations}
\label{sec:CRIO}
In this section, we show the implementation of controlled remote operations via a graph state. The $2N+1$ participants $A_1,A_2,...,A_{2N+1}$ share the entangled graph state as resource. They cooperate to  realize the remote operations  $U=\otimes_{j=1}^N\exp{[i\a_j\s_{n_{O_j}}]}$ on $N$ unknown states for remote systems $O_j$ by applying LOCC. Here $\a$ is the angle of rotation only known by participants $A_2,A_3,...,A_{N+1}$ and unknown to others including the controller. Besides, the controller takes the responsibility to decide whether or not and when the implementation of remote operations should be done on each system. Thus our protocol shows advantage in terms of strong security and extensive applications.
In Sec. \ref{sec:tripartite}, we introduce the implementation of operation $U_C=\exp{[i\a\s_{n_C}]}$  on one remote system $C$ via a tripartite graph state. In Sec. \ref{sec:fivepartite}, we show the protocol via a five-partite graph state, where two operations are implemented on two remote systems respectively.  In this protocol, appropriate local Hadamard operations are employed. So it is different from the protocol via tripartite state. Then we generalize this protocol into the one via a $(2N+1)$-partite graph state ($N\geq2$). It is presented in Sec. \ref{sec:2n+1partite}.  

\subsection{The protocol via a tripartite graph state}
\label{sec:tripartite}
In this protocol, three participants Alice, Bob and Charlie locate in distributed places. Alice is the controller who decides whether and when the remote implementation of operation  $U_C=\exp[i\a\s_{n_C}]$ on system $C$ can be realized by Bob and Charlie. The diagram of this process is shown in FIG. \ref{figure:tripartite}. The three participants share the tripartite graph state
\begin{eqnarray}
\label{def:stateh3}
\ket{h_3}=&&CZ_{(a,b)}CZ_{(a,c)}\ket{+}^{\otimes 3}\nonumber\\
=&&\frac{1}{2\sqrt{2}}(\ket{000}+\ket{001}+\ket{010}+\ket{011}+\ket{100}-\ket{101}-\ket{110}+\ket{111})_{a,b,c},
\end{eqnarray}
where qubits $a,b,c$ belong to Alice, Bob and Charlie, respectively. The entangled state $\ket{h_3}$ is prepared by the circuit shown in FIG.  \ref{fig:entstate31}.

	\begin{figure}[h!] 
	\centerline{
		\Qcircuit @C=0.8em @R=0.75em {
			\lstick{\ket{0}_a} & \gate{H} &\ctrl{1}&\ctrl{2}&\qw \\
			\lstick{\ket{0}_b} & \gate{H} &\gate{Z}&\qw&\qw \\
			\lstick{\ket{0}_c} & \gate{H} &\qw     &\gate{Z}&\qw \\
		}
	}
	\caption{The preparation of  graph state $\ket{h_3}$. The expressions of Hadamard and controlled-Z gates used in this circuit is given in (\ref{def:Hgate}) and (\ref{def:CZgate}), respectively.}
	\label{fig:entstate31}
\end{figure}
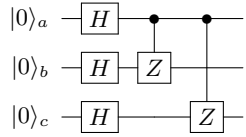

Given an unknown state $\ket{\Psi_{C}}$ for system $C$, Bob and Charlie implement the remote unitary operation $U_C=\exp[i\a\s_{n_C}]$ on system $C$, only with the permission of the controller Alice. This process is completed by applying LOCC via the shared entangled state $\ket{h_3}$.  Here $\a$ is any real number determined by Bob, and system $C$ belongs to Charlie. 

First, they construct the stator $S_3$, this process is performed on the state $\ket{\Psi_{C}}$. We do not show  $\ket{\Psi_{C}}$ in each step for simplicity. Charlie performs the local operation on his qubits $c$ and $C$, where
\begin{eqnarray}
U_{c,C}=\ket{0}_c\bra{0}\otimes I_C+\ket{1}_c\bra{1}\otimes \s_{n_C}.
\end{eqnarray} 
Here the operator $\s_{n_C}=\overline{n_C}\cdot\overline{\s}$  acting on the system $C$ satisfies $\s_{n_C}^2=I$. So the stator is 
\begin{eqnarray}
		\label{def:statorS3'}
S_3'=&&\frac{1}{2\sqrt{2}}(\ket{000}_{a,b,c}\otimes I_C+\ket{001}_{a,b,c}\otimes \s_{n_C}+\ket{010}_{a,b,c}\otimes I_C+\ket{011}_{a,b,c}\otimes \s_{n_C}\nonumber\\
+&&\ket{100}_{a,b,c}\otimes I_C-\ket{101}_{a,b,c}\otimes \s_{n_C}-\ket{110}_{a,b,c}\otimes I_C+\ket{111}_{a,b,c}\otimes \s_{n_C})\\
=&&\frac{1}{2}\ket{+}_a(\ket{00}_{b,c}\otimes I_C+\ket{11}_{b,c}\otimes \s_{n_C})+\frac{1}{2}\ket{-}_a(\ket{10}_{b,c}\otimes I_C+\ket{01}_{b,c}\otimes \s_{n_C}).\nonumber
\end{eqnarray}

 Now Alice's qubit $a$ is correlated with qubits $b$ and $c$. If Alice does not wish to cooperate with Bob and Charlie, she does nothing or something unknown to them. Then the relation between qubits $b$ and $c$ is unknown to Bob and Charlie and hence they can hardly construct the stator. Otherwise, if Alice wants the operations to be implemented, she performs a measurement of $\s_x$ on qubit $a$, and informs Bob of the measurement result. If the result is $\ket{-}$, Bob performs the operation $\s_x$ on qubit $b$, otherwise Bob need not perform any operation.  They obtain the stator
\begin{eqnarray}
S_3''=\ket{00}_{b,c}\otimes I_C+\ket{11}_{b,c}\otimes \s_{n_C}.
\end{eqnarray}
Then Charlie measures qubit $c$ in the basis $\ket{+}_c$ and $\ket{-}_c$. If the result is $\ket{-}$, Bob implements the operation $\s_{z}$ on qubit $b$, otherwise Bob does nothing.
Hence they construct the following stator $S_3$ successfully, 
\begin{eqnarray}
S_3=\ket{0}_b\otimes I_C+\ket{1}_b\otimes \s_{n_C}.
\end{eqnarray}

The second step is to implement the operation $U_C=\exp[i\a\s_{n_C}]$ on system $C$, with the help of  eigenoperator equation
\begin{eqnarray}
	\label{eq:eigenoperator3}
	e^{i\a\s_{x_b}}S_3=e^{i\a\s_{n_C}}S_3.
\end{eqnarray}
Note that $\s_{n_C}^2=I$, so we have  $e^{i\a\s_{n_C}}=\cos\a I_C+i\sin\a\s_{n_C}$.
 
 Now Bob and Charlie implement the remote operation by stator $S_3$. Bob implements the operation $\exp[i\a\s_{x_b}]$ on qubit $b$. Based on Eq. (\ref{eq:eigenoperator3}), the resulting state is 
\begin{eqnarray}
\ket{s}=&&e^{i\a\s_{x_b}}S_3\ket{\Psi_{C}}\\
=&&(\ket{0}_b\otimes I_C+\ket{1}_b\otimes \s_{n_C})e^{i\a\s_{n_C}}\ket{\Psi_{C}}.\nonumber
\end{eqnarray}
Bob measures qubit $b$ in the Z-basis $\ket{0_b}$ and $\ket{1_b}$, and informs Charlie the result. If it is $\ket{1_a}$, Charlie performs the local rotation $\exp[i\pi\s_{n_C}/2]=i\s_{n_C}$ on system $C$. This completes our CRIO protocol via $\ket{h_3}$.

\subsection{The protocol via a five-partite graph state}
\label{sec:fivepartite}
In this protocol, we take the five-partite graph state as entanglement resource. The controller Alice in this protocol decides whether and when the work of the following two groups can be done. The first group consists of Bob and David who aim to realize the remote implementation of operations $U_D=\exp{[i\a\s_{n_D}]}$, and the second group include Charlie and Eve who aim to realize $U_E=\exp{[i\b\s_{n_E}]}$.  They share the following graph state $\ket{h_5}$. For convenience of the measurement on qubit $a$, we rearrange the state in the basis $\ket{+}_a$ and $\ket{-}_a$,
\begin{eqnarray}
	\label{def:stateh5}
	\ket{h_5}=&&CZ_{(c,e)}CZ_{(c,d)}CZ_{(b,c)}CZ_{(a,c)}CZ_{(a,b)}\ket{+}^{\otimes 5}\nonumber\\
=	&&\frac{1}{4\sqrt{2}}(\ket{00000}+\ket{00100}+\ket{10000}+\ket{10100}+\ket{00001}-\ket{00101}+\ket{10001}-\ket{10101}\nonumber\\
	+&&\ket{01010}+\ket{01110}+\ket{11010}+\ket{11110}+\ket{01011}-\ket{01111}+\ket{11011}-\ket{11111}\nonumber\\
	+&&\ket{01000}-\ket{01100}-\ket{11000}+\ket{11100}+\ket{01001}+\ket{01101}-\ket{11001}-\ket{11101}\\
	+&&\ket{00010}-\ket{00110}-\ket{10010}+\ket{10110}+\ket{00011}+\ket{00111}-\ket{10011}-\ket{10111})_{a,b,c,d,e}\nonumber\\
	=&&\frac{1}{4}\ket{+}_a(\ket{0000}+\ket{0100}+\ket{0001}-\ket{0101}+\ket{1010}+\ket{1110}+\ket{1011}-\ket{1111})_{b,c,d,e}\nonumber\\
	+&&\frac{1}{4}\ket{-}_a(\ket{1000}-\ket{1100}+\ket{1001}+\ket{1101}+\ket{0010}-\ket{0110}+\ket{0011}+\ket{0111})_{b,c,d,e}.\nonumber
\end{eqnarray}
where qubits $a,b,c,d,e$ belong to Alice, Bob, Charlie, David, and Eve respectively, and the preparation of this state is shown in FIG. \ref{fig:entstate51}. For the operations  $U_D=\exp{[i\a\s_{n_D}]}$ and $U_E=\exp{[i\b\s_{n_E}]}$, $\a$ and $\b$ are real numbers determined by Bob and Charlie, respectively. System $D$ and $E$ belong to David and Eve respectively.

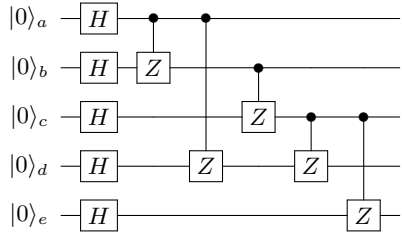
\begin{figure}[h!] 
	\centerline{
		\Qcircuit @C=0.8em @R=0.75em {
			\lstick{\ket{0}_a} & \gate{H}&\ctrl{1} &\ctrl{3}&\qw     &\qw     &\qw     &\qw\\
			\lstick{\ket{0}_b} & \gate{H}&\gate{Z} &\qw     &\ctrl{1}&\qw     &\qw     &\qw\\
			\lstick{\ket{0}_c} & \gate{H}&\qw      &\qw     &\gate{Z}&\ctrl{1}&\ctrl{2}&\qw\\
			\lstick{\ket{0}_d} & \gate{H}&\qw      &\gate{Z}&\qw     &\gate{Z}&\qw     &\qw\\	
			\lstick{\ket{0}_e} & \gate{H}&\qw      &\qw     &\qw     &\qw     &\gate{Z}&\qw\\
		}
	}
	\caption{The preparation of graph state $\ket{h_5}$. The corresponding graph for this state is shown in FIG. \ref{fig:graph}. }
	\label{fig:entstate51}
\end{figure}

 The first step is to construct the stator.  For convenience, we ignore the global factor in the following statement.  David and Eve perform the local operations $U_{d,D}$ and $U_{e,E}$ respectively, where
 \begin{eqnarray}
 U_{d,D}=&&\ket{0}_d\bra{0}\otimes I_D+\ket{1}_d\bra{1}\otimes \s_{n_D},\\
 U_{e,E}=&&\ket{0}_e\bra{0}\otimes I_E+\ket{1}_e\bra{1}\otimes \s_{n_E}.
 \end{eqnarray} 
This yields the stator 
 \begin{eqnarray}
 S'_5=&&\ket{+}_a(\ket{0000}_{b,c,d,e}\otimes I_DI_E+\ket{0100}_{b,c,d,e}\otimes I_DI_E+\ket{0001}_{b,c,d,e}\otimes I_D\s_{n_E}-\ket{0101}_{b,c,d,e}\otimes I_D\s_{n_E}\nonumber\\
 +&&\ket{1010}_{b,c,d,e}\otimes \s_{n_D}I_E+\ket{1110}_{b,c,d,e}\otimes \s_{n_D}I_E+\ket{1011}_{b,c,d,e}\otimes\s_{n_D}\s_{n_E}-\ket{1111}_{b,c,d,e}\otimes\s_{n_D}\s_{n_E})\nonumber\\
 +&&\ket{-}_a(\ket{1000}_{b,c,d,e}\otimes I_DI_E-\ket{1100}_{b,c,d,e}\otimes I_DI_E+\ket{1001}_{b,c,d,e}\otimes I_D\s_{n_E}+\ket{1101}_{b,c,d,e}\otimes I_D\s_{n_E}\nonumber\\
 +&&\ket{0010}_{b,c,d,e}\otimes \s_{n_D}I_E-\ket{0110}_{b,c,d,e}\otimes \s_{n_D}I_E+\ket{0011}_{b,c,d,e}\otimes \s_{n_D}\s_{n_E}+\ket{0111}_{b,c,d,e}\otimes\s_{n_D}\s_{n_E}).\nonumber\\
 \end{eqnarray}
 Next, Charlie performs the operation $H$ on qubit $c$. Note that $H\ket{+}=\ket{0}$ and $H\ket{-}=\ket{1}$. The stator becomes
  \begin{eqnarray}
  	\label{def:S''5}
 	S''_5=&&\ket{+}_a(\ket{0000}_{b,c,d,e}\otimes I_DI_E+\ket{0101}_{b,c,d,e}\otimes I_D\s_{n_E}
 	+\ket{1010}_{b,c,d,e}\otimes \s_{n_D}I_E+\ket{1111}_{b,c,d,e}\otimes\s_{n_D}\s_{n_E})\nonumber\\
 	+&&\ket{-}_a(\ket{1100}_{b,c,d,e}\otimes I_DI_E+\ket{1001}_{b,c,d,e}\otimes I_D\s_{n_E}+\ket{0110}_{b,c,d,e}\otimes \s_{n_D}I_E+\ket{0011}_{b,c,d,e}\otimes \s_{n_D}\s_{n_E}).\nonumber\\
 \end{eqnarray}
From (\ref{def:S''5}), one can find that Alice's qubit $a$ is correlated with qubits $c$ and $e$. If Alice does not wish to cooperate with the two groups, she does nothing or something unknown to them. Then the relation between qubits $b$ and $d$ is unknown to Bob and David and hence they can hardly construct the stator. So it does with Charlie and Eve. Otherwise, Alice wants the operations to be implemented then she implements the measurement of $\s_{x}$ on qubit $a$ and informs Bob and Charlie of the result by classical communication. If the result is $\ket{-}_a$, Bob and Charlie perform the operation $\s_{x}$ on qubit $b$, otherwise they do nothing. Thus they obtain the stator
\begin{eqnarray}
S'''_5=\ket{0000}_{b,c,d,e}\otimes I_DI_E+\ket{0101}_{b,c,d,e}\otimes I_D\s_{n_E}+\ket{1010}_{b,c,d,e}\otimes \s_{n_D}I_E+\ket{1111}_{b,c,d,e}\otimes \s_{n_D}\s_{n_E}.\nonumber\\
\end{eqnarray} 
David and Eve perform the measurement of $\s_{x}$ on qubit $d$ and $e$, and send their measurement results to Bob and Charlie, respectively. The diagram of this process is shown in FIG. \ref{fig:measurementde}.

\begin{figure}[h]
	\center{\includegraphics[width=8cm]  {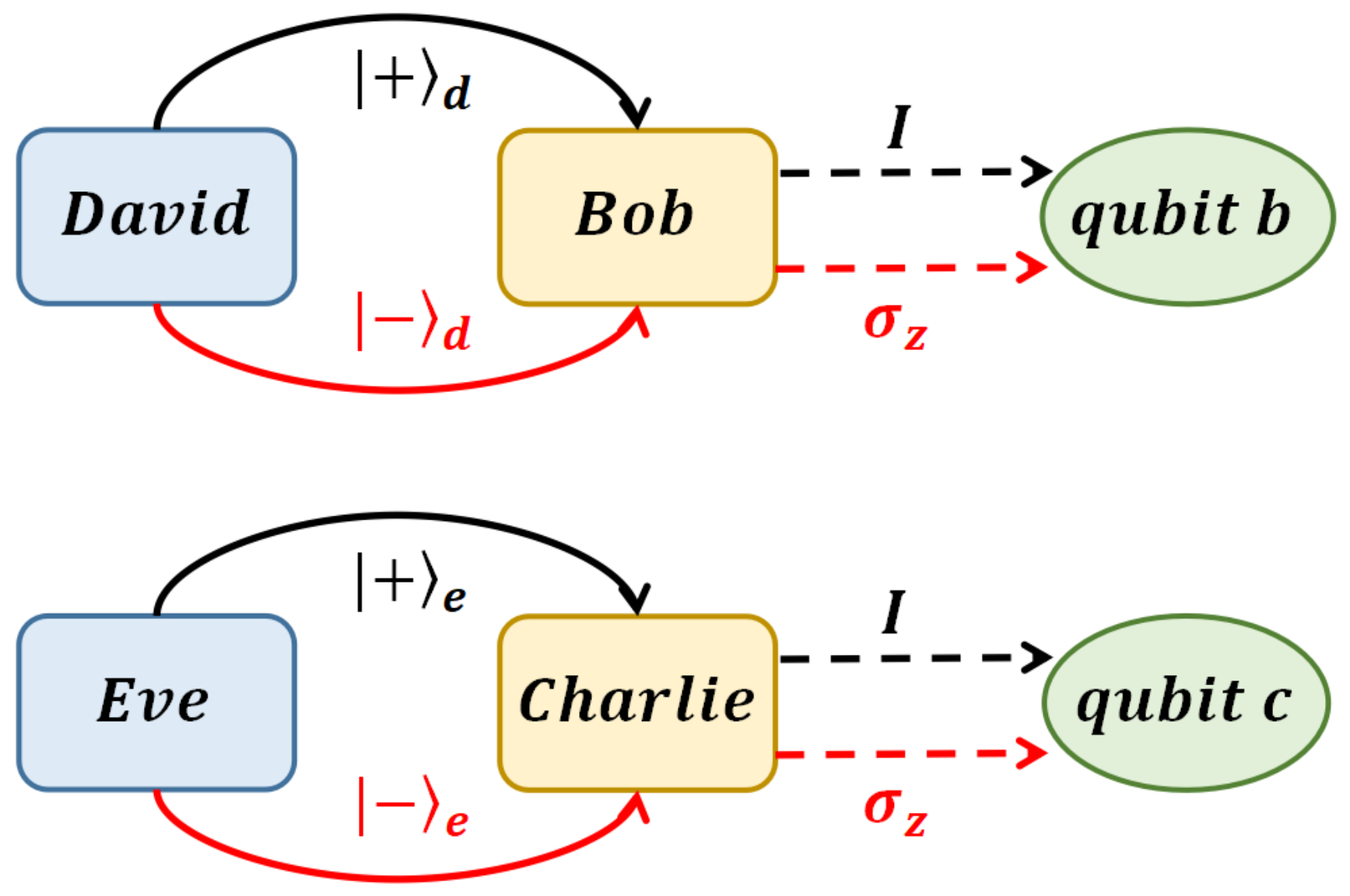}}
	\caption{Diagram showing the process for the measurement of qubits $d$ and $e$. David and Eve perform the measurement of $\s_{x}$ on qubit $d$ and $e$, respectively. The solid lines represent the measurement result, and the dashed lines represent the corresponding operations that performed on specified qubits.   When receiving the result $\ket{-}$, Bob and Charlie perform the operation $\s_{z}$ on their qubits, otherwise they do nothing($I$). }
	\label{fig:measurementde}
\end{figure}

 So they prepare the stator $S_5$ successfully, where
 \begin{eqnarray}
 S_5=\ket{00}_{b,c}\otimes I_DI_E+\ket{01}_{b,c}\otimes I_D\s_{n_E}+\ket{10}_{b,c}\otimes \s_{n_D}I_E+\ket{11}_{b,c}\otimes \s_{n_D}\s_{n_E}.
 \end{eqnarray}

Next they implement the operations  $\exp{[i\a\s_{n_D}]}$ and $\exp{[i\b\s_{n_E}]}$ on two unknown states $\ket{\Phi_D}$ and $\ket{\Phi_E}$  for system $D$ and $E$, respectively.
The relation $\s_{x_b}\s_{x_c}S_5=\s_{n_D}\s_{n_E}S_5$ can be established. Hence the stator $S_5$ also satisfies the following eigenoperator equation, which is similar to Eq. (\ref{eq:eigenoperator3}),  
 \begin{eqnarray}
 	\label{eq:eigenoperator5}
 e^{i\a\s_{x_b}}e^{i\b\s_{x_c}}S_5=e^{i\a\s_{n_D}}e^{i\b\s_{n_E}}S_5.
 \end{eqnarray}
Bob and Charlie implement the operations $\exp{[i\a\s_{x_b}]}$ and $\exp{[i\b\s_{x_c}]}$ on their qubits $b$ and $c$ respectively. Based on Eq. (\ref{eq:eigenoperator5}), we obtain that the state is now
\begin{eqnarray}
 &&e^{i\a\s_{x_b}}e^{i\b\s_{x_c}}S_5\ket{\Phi_D}\ket{\Phi_E}\nonumber\\
 =&&(\ket{00}_{b,c}\otimes I_DI_E+\ket{01}_{b,c}\otimes I_D\s_{n_E}+\ket{10}_{b,c}\otimes \s_{n_D}I_E\\
 +&&\ket{11}_{b,c}\otimes \s_{n_D}\s_{n_E})e^{i\a\s_{n_D}}\ket{\Phi_D}e^{i\b\s_{n_E}}\ket{\Phi_E}.\nonumber
\end{eqnarray}
Finally Bob and Charlie perform the measurement of $\s_{z}$ on qubits $b$ and $c$. Bob transmits the measurement result to David, and Charlie's result sends to Eve. If the sender's measurement result is $\ket{1}_b(\ket{1}_c)$, the receiver performs the operation $\exp{[i\pi\s_{n_D}/2]}=i\s_{n_D}$ ($\exp{[i\pi\s_{n_E}/2]}=i\s_{n_E}$), otherwise the receiver need not perform any operation.

\subsection{The protocol via a $(2N+1)$-partite graph state}
\label{sec:2n+1partite}
We show the protocol via a $(2N+1)$-partite graph state, which is the generalization of the protocol in Sec. \ref{sec:fivepartite}. In this protocol, $2N+1$ participants $A_1, A_2,..., A_{2N+1}$ are involved. Participant $A_1$ is the controller who supervises  $N$ distributed groups. Participants $A_{j-N}$ and $A_{j}$ work as a group to  implement the operation $\exp{[i\b_j\s_{n_{O_j}}]}$ on the unknown states $\ket{\Psi_j}$ for system $O_j$, for $j=N+2,N+3,...,2N+1$. For each group, the real numbers $\b_j$ and $\s_{n_{O_{j}}}$ is only available to  participants $A_{j-N}$ and $A_j$, respectively. During this process, the controller $A_1$  determines whether and when these operations should be implemented, but knows nothing about the parameters of the operations. That is to say, none of the participants is able to obtain the complete information of the target operations.  It guarantees the security of our protocol.  The $2N+1$ participants share the following $(2N+1)$-qubit graph state,
\begin{eqnarray}
	\label{def:stateh2n+1}
\ket{h_{2N+1}}=&&CZ_{(1,2)} CZ_{(1,N+2)}\prod_{k=3}^{N+1} \left[CZ_{(2,k)}CZ_{(k,N+2)} CZ_{(k,k+N)}\right]\ket{+}^{\otimes (2N+1)}\\
=&&\frac{1}{2^N\sqrt{2}}\sum_{q_1,q_2,...,q_{2N+1}=0}^1(-1)^{f(x)}\ket{q_1,q_2,...,q_{2N+1}}_{a_1,a_2,...,a_{2N+1}},\nonumber
\end{eqnarray}
where qubit $a_k$ belongs to participant  $A_k$ respectively, for $k=1,2,...,2N+1$. The boolean function $f(x):\{0,1\}^{2N+1}\rightarrow\{0,1\}$, for $x\equiv q_1q_2...q_{2N+1}$ and $q_j\in\{0,1\},j=1,2,...,2N+1$,
\begin{eqnarray}
f(x)=q_1q_2\oplus q_1q_{N+2}\oplus_{k=3}^{N+1} (q_2q_k\oplus q_kq_{N+2}\oplus q_{k}q_{k+N}).
\end{eqnarray}
Here $"\oplus"$ denotes plus modulo two, and 
\begin{eqnarray}
	q_sq_t=
	\begin{cases}
		1&\mbox{for $q_s=q_t=1$,}\\
		0&\mbox{otherwise.}
	\end{cases}
\end{eqnarray}
The state $\ket{h_{2N+1}}$ is prepared by the quantum circuit shown in FIG. \ref{fig:entstate2N+1}.

\begin{figure}[htp] 
	\centerline{
		\Qcircuit @C=0.8em @R=0.75em {
\lstick{\ket{0}_{a_1}}    & \gate{H}&\ctrl{1} &\ctrl{6}&\qw     &\qw     &\qw &\qw     &\qw     &\qw     &\qw     &\qw     &\qw&\qw  &\qw&\qw     &\qw     &\qw&\qw&\qw\\
\lstick{\ket{0}_{a_2}}    & \gate{H}&\gate{Z} &\qw     &\ctrl{1}&\ctrl{2}&\qw &\ctrl{4}&\qw     &\qw     &\qw     &\qw     &\qw&\qw  &\qw&\qw     &\qw     &\qw&\qw&\qw\\
\lstick{\ket{0}_{a_3}}    & \gate{H}&\qw      &\qw     &\gate{Z}&\qw     &\qw &\qw     &\ctrl{4}&\ctrl{5}&\qw     &\qw     &\qw&\qw  &\qw&\qw     &\qw     &\qw&\qw&\qw\\
\lstick{\ket{0}_{a_4}}    & \gate{H}&\qw      &\qw     &\qw     &\gate{Z}&\qw &\qw     &\qw     &\qw     &\ctrl{3}&\ctrl{5}&\qw&\qw  &\qw&\qw     &\qw     &\qw&\qw&\qw\\	
\lstick{\vdots}           & \vdots  &         &        &        &        &\cdots&      &        &        &        &        &   &\cdots&  &        &        &\vdots&\\       
\lstick{\ket{0}_{a_{N+1}}}& \gate{H}&\qw      &\qw     &\qw     &\qw     &\qw &\gate{Z}&\qw     &\qw     &\qw     &\qw     &\qw&\qw  &\qw&\ctrl{1}&\ctrl{5}&\qw&\qw&\qw\\
\lstick{\ket{0}_{a_{N+2}}}& \gate{H}&\qw      &\gate{Z}&\qw     &\qw     &\qw &\qw     &\gate{Z}&\qw     &\gate{Z}&\qw     &\qw&\qw  &\qw&\gate{Z}&\qw     &\qw&\qw&\qw\\
\lstick{\ket{0}_{a_{N+3}}}& \gate{H}&\qw      &\qw     &\qw     &\qw     &\qw &\qw     &\qw     &\gate{Z}&\qw     &\qw     &\qw&\qw  &\qw&\qw     &\qw     &\qw&\qw&\qw\\
\lstick{\ket{0}_{a_{N+4}}}& \gate{H}&\qw      &\qw     &\qw     &\qw     &\qw &\qw     &\qw     &\qw     &\qw     &\gate{Z}&\qw&\qw  &\qw&\qw     &\qw     &\qw&\qw&\qw\\
\lstick{\vdots}           &\vdots   &         &        &        &        &    &        &        &        &        &        &   &\cdots&  &        &        &\vdots&   \\ 
\lstick{\ket{0}_{a_{2N+1}}}&\gate{H}&\qw      &\qw     &\qw     &\qw     &\qw &\qw     &\qw     &\qw     &\qw     &\qw     &\qw&\qw  &\qw&\qw     &\gate{Z}&\qw&\qw&\qw\\
		}
	}
	\caption{The preparation of graph state $\ket{h_{2N+1}}$ shown in (\ref{def:stateh2n+1}). The controlled-Z gates $CZ(1,2)$, $CZ(1,N+2)$, $CZ(2,j)$, $CZ(j,N+2)$, $CZ(j,j+N)$ are employed in this circuit, for $j=3,4,...,N+1$. They correspond to edges in the second graph shown in FIG. \ref{fig:graph}. }
	\label{fig:entstate2N+1}
\end{figure}
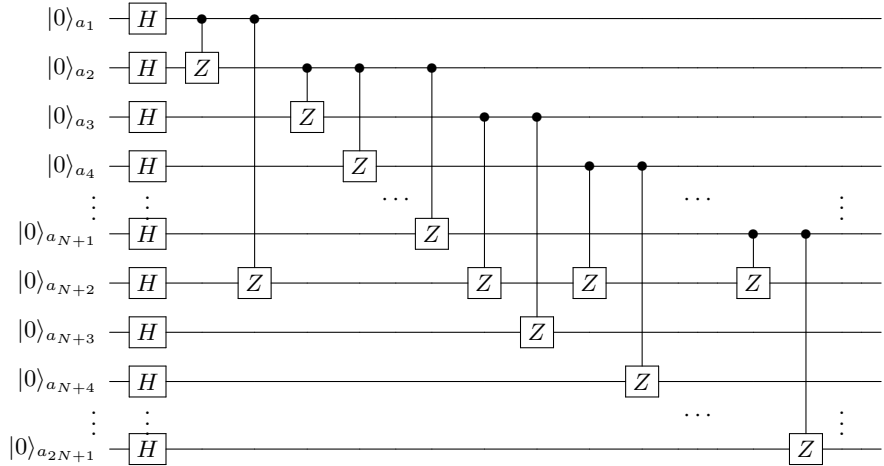

\begin{figure}[htbp]
	\centering
	\subfigure{
		\includegraphics[width=5.5cm]{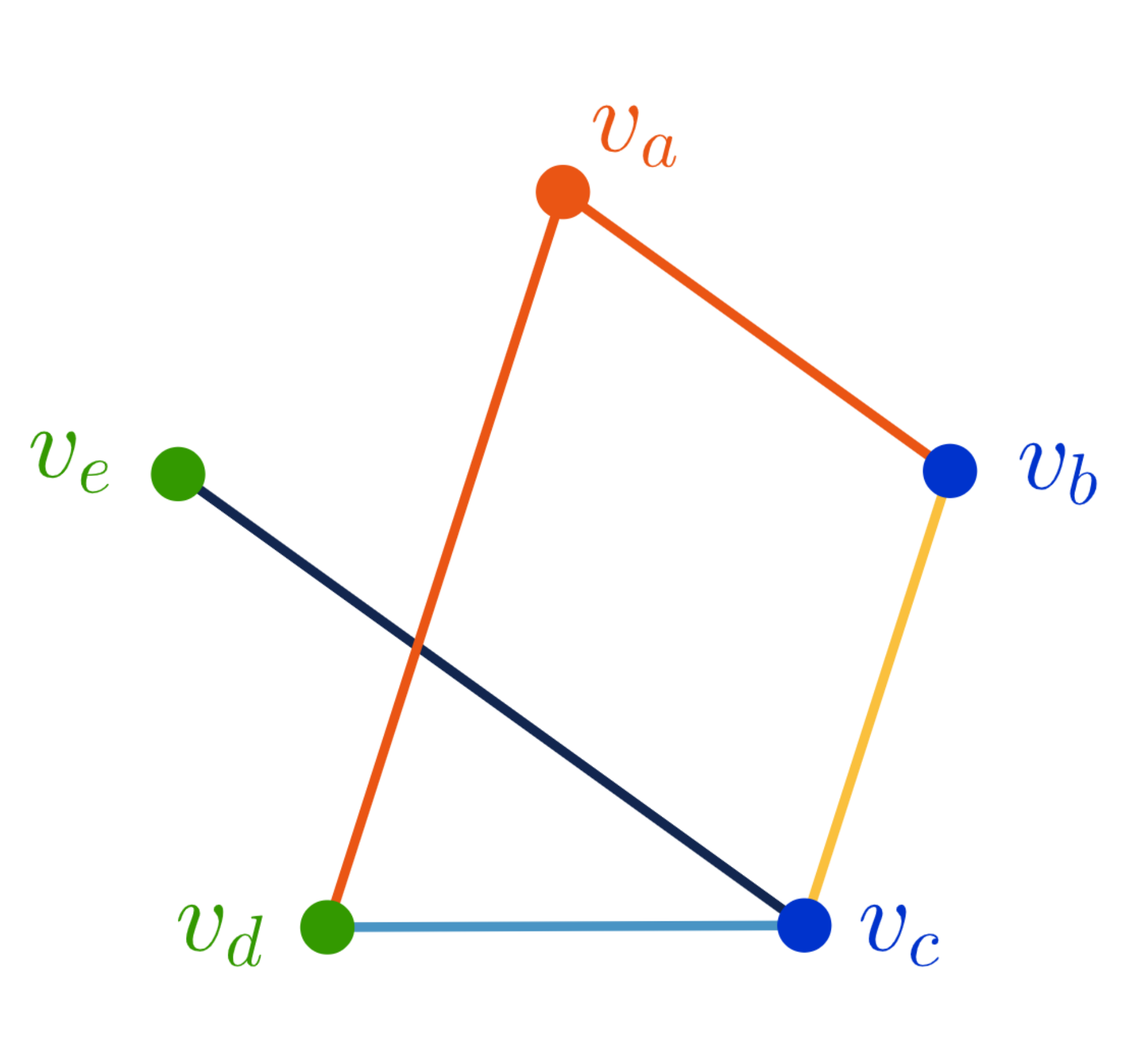}
	}
\quad
	\subfigure{
		\includegraphics[width=6.6cm]{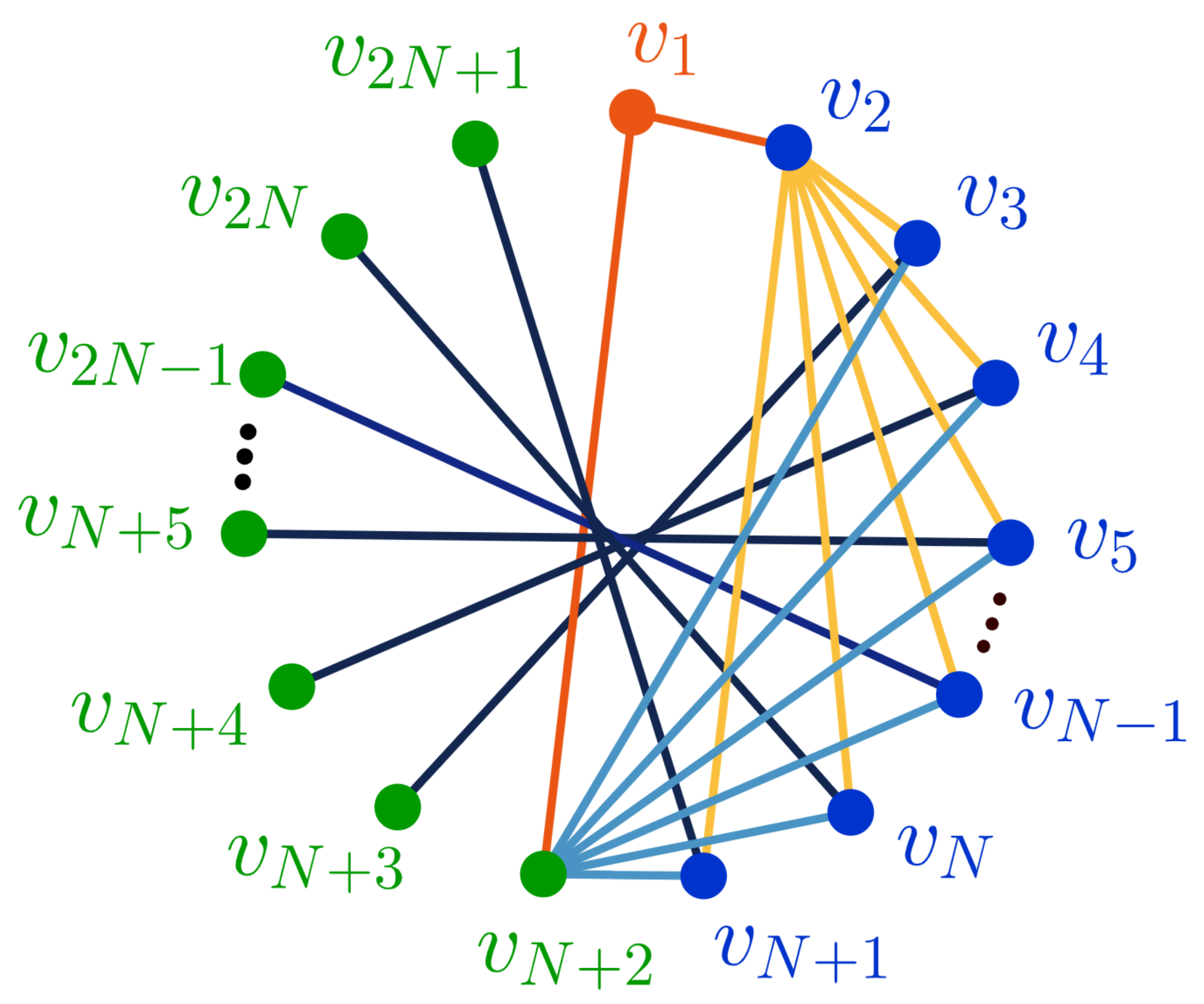}
	}
	\caption{The graphs $G_5=(V_5,E_5)$ and $G_{2N+1}=(V_{2N+1},E_{2N+1})$ correspond to graph states $\ket{h_5}$ in (\ref{def:stateh5}) and $\ket{h_{2N+1}}$ in (\ref{def:stateh2n+1}), respectively. Here $V_5=\{v_a,v_b,v_c,v_d,v_e\}$, $E_5=\{\{v_a,v_b\}, \{v_a,v_d\}, \{v_b,v_c\}, \{\{v_c,v_d\}, \{\{v_c,v_e\}\}$ and $V_{2N+1}=\{v_1,v_2,...,v_{2N+1}\}$, $E_{2N+1}=\{\{v_1,v_2\}, \{v_1,v_{N+2}\}, \{v_2,v_{j}\}, \{v_{j}, v_{N+2}\}, \{v_{j}, v_{N+j}\}\}$ with $j=3,4,...,N+1$. Each vertex is associated to a participant in the protocol.  Green vertices $v_{N+2},v_{N+3},...,v_{2N+1}$ correspond to participants owning $N$ remote systems, and the blue ones $v_2,v_3,...,v_{N+1}$ correspond to the participants aiming to implement remote operations on these systems.}
	\label{fig:graph}
\end{figure}

Now we show the realization of the remote operations by the cooperation of $2N+1$ participants. First, they prepare the stator $S_{2N+1}$ by LOCC, only with the permission of controller $A_1$. The stators are constructed on the unknown states $\ket{\Psi_j}$ for remote systems $O_j$. Second, the operations are implemented on these systems with the help of the stator. The construction of stator consists of four steps, i.e. STEP 1-4 in the following statement. In each step, the participants perform a kind of local operations or measurements. For convenience, we ignore the global factor in the following statement.

STEP 1: The participants $A_j$ perform the local operations 
\begin{eqnarray}
 U_{a_j,O_j}=\ket{0}_{a_j}\bra{0}\otimes I_{O_j}+\ket{1}_{a_j}\bra{1}\otimes \s_{n_{O_j}}
\end{eqnarray}
 on the entangled state $\ket{h_{2N+1}}$ in (\ref{def:stateh2n+1}) for qubit $a_j$ and system $O_j$ respectively, where $j=N+2,N+3,...,2N+1$. The operation $\s_{n_{O_j}}$ satisfies that $\s^0_{n_{O_j}}=\s^2_{n_{O_j}}=I_{O_j}$. The resulting stator is 
\begin{eqnarray}
	\label{def:statorn1}
S'_{2N+1}=\sum_{q_1,q_2,...,q_{2N+1}=0}^1&&(-1)^{f(x)}\ket{q_1,q_2,...,q_{2N+1}}_{a_1,a_2,...,a_{2N+1}}\\
&&\otimes\s^{q_{N+2}}_{n_{O_{N+2}}}\otimes  \s^{q_{N+3}}_{n_{O_{N+3}}}\otimes...\otimes\s^{q_{2N+1}}_{n_{O_{2N+1}}}.\nonumber
\end{eqnarray}

For convenience of the statement, we rearrange the stator $S'_{2N+1}$ according to qubit $a_1$ in the basis $\ket{+}_{a_1}$ and $\ket{-}_{a_1}$ as follows, 
\begin{eqnarray}
	\label{def:statorn11}
S'_{2N+1}=&&\ket{+}_{a_1}\sum_{q_2,...,q_{2N+1}}[(-1)^{g_1(x')}+(-1)^{g_2(x')}]\ket{q_2,...,q_{2N+1}}\otimes\s^{q_{N+2}}_{n_{O_{N+2}}}\otimes  ...\otimes\s^{q_{2N+1}}_{n_{O_{2N+1}}}\nonumber\\
+&&\ket{-}_{a_1}\sum_{q_2,...,q_{2N+1}}[(-1)^{g_1(x')}-(-1)^{g_2(x')}]\ket{q_2,...,q_{2N+1}}\otimes \s^{q_{N+2}}_{n_{O_{N+2}}}\otimes...\otimes\s^{q_{2N+1}}_{n_{O_{2N+1}}},
\end{eqnarray}
where $g_1(x')$ and $g_2(x')$ are boolean functions of $x'=q_2q_3...q_{2N+1}$:
\begin{eqnarray}
	\label{def:functiong_1}
g_1(x')=&&\oplus_{k=3}^{N+1}(q_2q_{k}\oplus q_kq_{N+2}\oplus q_{k}q_{k+N}),\\
\label{def:functiong_2}
g_2(x')=&&q_2\oplus q_{N+2}\oplus_{k=3}^{N+1}(q_2q_{k}\oplus q_kq_{N+2}\oplus q_{k}q_{k+N}).
\end{eqnarray} 
From Eqs. (\ref{def:functiong_1}) and (\ref{def:functiong_2}), one can obtain that $g_1(x')=g_2(x')$ when $q_2=q_{N+2}$, and $g_1(x')=g_2(x')\oplus 1$ when $q_2=q_{N+2}\oplus 1$. Hence, ignoring the global factor, the stator in (\ref{def:statorn11}) is equal to
\begin{eqnarray}
S'_{2N+1}=\ket{+}_{a_1}\otimes M+\ket{-}_{a_1}\otimes N,
\end{eqnarray}
where
\begin{eqnarray}
	\label{def:statorM}
M=\sum_{\substack{q_2,q_3,...,q_{N+1}\\q_{N+3},...,q_{2N+1}}}\!\!\!\!&&(-1)^{h_1(x'')}\ket{q_2,q_3,...,q_{N+1},q_{2},q_{N+3},...,q_{2N+1}}_{a_2,a_3,...,a_{2N+1}}\\&&\otimes\s^{q_{2}}_{n_{O_{N+2}}}\otimes \s^{q_{N+3}}_{n_{O_{N+3}}}\otimes ...\otimes\s^{q_{2N+1}}_{n_{O_{2N+1}}},\nonumber\\
\label{def:statorN}
N=\sum_{\substack{q_2,q_3,...,q_{N+1}\\q_{N+3},...,q_{2N+1}}}\!\!\!\!&&(-1)^{h_2(x'')}\ket{q_2,q_3,...,q_{N+1},q_{2}\oplus1,q_{N+3},...,q_{2N+1}}_{a_2,a_3,...,a_{2N+1}}\\&&\otimes\s^{q_{2}\oplus 1}_{n_{O_{N+2}}}\otimes \s^{q_{N+3}}_{n_{O_{N+3}}}\otimes ...\otimes\s^{q_{2N+1}}_{n_{O_{2N+1}}}.\nonumber
\end{eqnarray}
Here $h_1(x'')$ and $h_2(x'')$ are boolean functions of $x''=q_3q_4...q_{2N+1}$:
\begin{eqnarray}
h_1(x'')=&&\oplus_{k=3}^{N+1}q_kq_{k+N},\\
h_2(x'')=&&\oplus_{k=3}^{N+1}(q_k\oplus q_kq_{k+N}).
\end{eqnarray}

STEP 2: Participants $A_3, A_4,...,A_{N+1}$ perform the operation $H$ on their qubits $a_3,a_4,...,a_{N+1}$, respectively.

To show the effect of the operation $\otimes_{k=3}^{N+1}H_{a_k}$  on the stator $S'_{2N+1}$ in (\ref	{def:statorn11}) clearly, we rearrange $M$  in (\ref{def:statorM}) and $N$ in (\ref{def:statorN}) as follows,
\begin{eqnarray}
	\label{def:statorM2}
M=\sum_{q_2,q_{N+3},...,q_{2N+1}}\!\!\!\!&&\ket{q_2}_{a_2}\otimes_{k=3}^{N+1}\left(\sum_{q_k}(-1)^{q_kq_{k+N}}\ket{q_k}_{a_k}\right)\otimes\ket{q_2}_{a_{N+2}}\otimes_{m=3}^{N+1}\ket{q_{m+N}}_{a_{m+N}}\\&&\otimes\s^{q_{2}}_{n_{O_{N+2}}}\otimes \s^{q_{N+3}}_{n_{O_{N+3}}}\otimes ...\otimes\s^{q_{2N+1}}_{n_{O_{2N+1}}},\nonumber\\
\label{def:statorN2}
N=\sum_{q_2,q_{N+3},...,q_{2N+1}}\!\!\!\!&&\ket{q_2}_{a_2}\otimes_{k=3}^{N+1}\left(\sum_{q_k}(-1)^{q_k\oplus q_kq_{k+N}}\ket{q_k}_{a_k}\right)\otimes\ket{q_2\oplus 1}_{a_{N+2}}\otimes_{m=3}^{N+1}\ket{q_{m+N}}_{a_{m+N}}\\&&\otimes\s^{q_{2}\oplus1}_{n_{O_{N+2}}}\otimes \s^{q_{N+3}}_{n_{O_{N+3}}}\otimes ...\otimes\s^{q_{2N+1}}_{n_{O_{2N+1}}}.\nonumber
\end{eqnarray}

 For $k=3,4,...,N+1$, it holds that
\begin{eqnarray}
	\label{def:Hqn+k}
&&H_{a_k}(\sum_{q_k=0}^{1}(-1)^{q_kq_{k+N}}\ket{q_k}_{a_k})=\sqrt{2}\ket{q_{k+N}}_{a_k},\\
\label{def:Hqn+k+1}
&&H_{a_k}(\sum_{q_k=0}^{1}(-1)^{q_k\oplus q_kq_{k+N}}\ket{q_k}_{a_k})=\sqrt{2}\ket{q_{k+N}\oplus1}_{a_k}.
\end{eqnarray}
Applying (\ref{def:Hqn+k}), (\ref{def:Hqn+k+1}) to (\ref{def:statorM2}), (\ref{def:statorN2}) respectively, it yields the  following stator $S''_{2N+1}$.  The global factor $\sqrt{2}$ is ignored here, 
\begin{eqnarray}
	\label{def:S''2N+1}
	S''_{2N+1}=&&\ket{+}_{a_1}\sum_{q_2,q_3,...,q_{N+1}}\ket{q_2,...,q_{N+1},q_{2},...,q_{N+1}}\otimes\s^{q_{2}}_{n_{O_{N+2}}}\otimes\s^{q_{3}}_{n_{O_{N+3}}}\otimes  ...\otimes\s^{q_{N+1}}_{n_{O_{2N+1}}}\\
	+&&\ket{-}_{a_1}\sum_{q_2,q_3,...,q_{N+1}}\ket{q_2,...,q_{N+1},q_{2}\oplus1,...,q_{N+1}\oplus 1}\otimes\s^{q_{2}\oplus1}_{n_{O_{N+2}}}\otimes\s^{q_{3}\oplus1}_{n_{O_{N+3}}}\otimes  ...\otimes\s^{q_{N+1}\oplus1}_{n_{O_{2N+1}}}.\nonumber
\end{eqnarray}

STEP 3: From (\ref{def:S''2N+1}), one can find that the controller's qubit $a_1$ is correlated with qubits $a_k$ and $a_{k+N}$, for $k=2,3,...,N+1$. If the controller $A_1$ does not wish to cooperate with other groups including $A_k$ and $A_{k+N}$, he does nothing or something unknown to other groups. Then the relation between qubits $a_k$ and $a_{k+N}$ is unknown to participants $A_k$ and $A_{k+N}$ and hence they cannot realize the operations at the beginning of this subsection. Otherwise, if the controller $A_1$ permits the realization of these operations, he implements the measurement of $\s_{x}$ on qubit $a_1$, and informs $A_2,A_3,...,A_{N+1}$ of the result. If the result is $\ket{-}$, the receivers $A_2,A_3,...,A_{N+1}$ perform $\s_{x}$ on their qubits, otherwise they do nothing.  This yields the stator 
\begin{eqnarray}
S'''_{2N+1}=\sum_{q_2,q_3...,q_{N+1}}\!\!\!&&\ket{q_2,q_3...,q_{N+1},q_{2},q_3,...,q_{N+1}}_{a_2,a_3,...,a_{2N+1}}\\
&&\otimes\s^{q_{2}}_{n_{O_{N+2}}}\otimes\s^{q_3}_{n_{O_{N+3}}}\otimes  ...\otimes\s^{q_{N+1}}_{n_{O_{2N+1}}}\nonumber.
\end{eqnarray}

STEP 4: Participants $A_{N+2}$, $A_{N+3}$,..., $A_{2N+1}$ measure their qubits in the X-basis, and inform $A_2, A_3,..., A_{N+1}$ of the results by classical communication, respectively. If the sender's result is $\ket{-}$, the receiver performs the operation $\s_{z}$ on his qubit, otherwise the receiver does nothing. So they construct the stator
\begin{eqnarray}
S_{2N+1}=\sum_{q_2,q_3,...,q_{N+1}}\!\!\! &&\ket{q_2,q_3,...,q_{N+1}}_{a_2,a_3,...,a_{N+1}}\\
&&\otimes\s^{q_2}_{n_{O_{N+2}}}\otimes\s^{q_3}_{n_{O_{N+3}}}\otimes...\otimes\s^{q_{N+1}}_{n_{O_{2N+1}}},\nonumber
\end{eqnarray}
where $\s^{0}_{n_{O_{j}}}=I_{O_{j}}, j=N+2,...,2N+1$.

Then the stator $S_{2N+1}$ is employed to implement the remote operations on the unknown state $\ket{\Psi_j}$ for remote systems $O_j$.  Similar to the eigenoperator equation shown in (\ref{eq:firststator}), the equation for stator $S_{2N+1}$ still holds,
 \begin{eqnarray}
	\label{eq:eigenoperator2n+1}
	\left(\bigotimes_{k=2}^{N+1}\exp{[i\b_{k+N}\s_{x_{a_{k}}}]}\right) S_{2N+1}
	=\left(\bigotimes_{j=N+2}^{2N+1}\exp{[i\b_{j}\s_{n_{O_{j}}}]}\right)S_{2N+1},
\end{eqnarray}
where $\s^2_{n_{O_{j}}}=I_{O_j}$ and $\b_{j}$ is any real number determined by the participant $A_{j-N}$. By performing the local operation $	\exp{[i\b_{j}\s_{x_{a_{j-N}}}]}$ on qubit $a_{j-N}$, the participant $A_{j-N}$ realizes the  corresponding operation  $\exp{[i\b_{j}\s_{n_{O_{j}}}]}$ on remote system $O_j$.
The implementation consists of two steps.

STEP 5: The participant $A_{k}$ performs the local operation $	\exp{[i\b_{k+N}\s_{x_{a_k}}]}$ on qubit $a_{k}$,  for $k=2,3,...,N+1$. Using Eq. (\ref{eq:eigenoperator2n+1}), the state is now
\begin{eqnarray}
S_{2N+1}\left(\bigotimes_{j=N+2}^{2N+1}\exp{[i\b_{j}\s_{n_{O_{j}}}]}\ket{\Psi_{j}}\right).
\end{eqnarray}

STEP 6: To eliminate the remaining stator $S_{2N+1}$ on the state $\ket{\Psi_{j}}$ and only keep the target operation, the participant $A_{k}$ implements the measurement of  $\s_{z}$ on qubit $a_k$, and informs $A_{k+N}$ of the measurement result, where $k=2,3,...,N+1$. If the sender $A_k$'s result is $\ket{1}$, the receiver $A_{k+N}$ performs the operation $\exp{[i\pi\s_{n_{O_{k+N}}}/2]}=i\s_{n_{O_{k+N}}}$, otherwise the receiver need not perform any operation. So they eliminate the stator $S_{2N+1}$. The resulting state is 
\begin{eqnarray}
\bigotimes_{j=N+2}^{2N+1}\exp{[i\b_{j}\s_{n_{O_{j}}}]}\ket{\Psi_{j}}.
\end{eqnarray}
It indicates that they have implemented the remote operation successfully. 
\begin{figure}[h]
	\center{\includegraphics[width=11cm]  {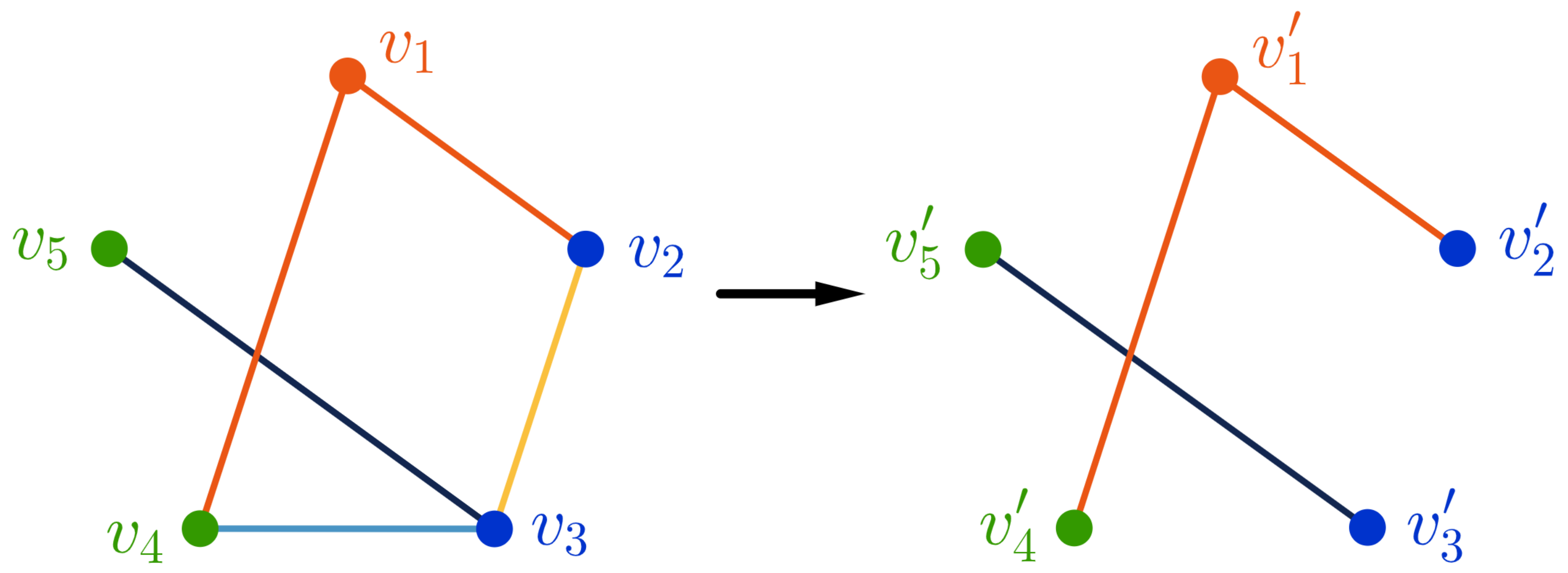}}
	\caption{The graph $G_5=(V_5,E_5)$ on the left hand side   corresponds to graph state $\ket{h_5}$, where $V_5=\{v_1,v_2,v_3,v_4,v_5\}$, $E_5=\{\{v_1,v_2\}, \{v_1,v_4\}, \{v_2,v_3\}, \{v_3,v_4\}, \{v_3,v_5\}\}$.  The other graph  $G'_5=(V'_{5},E'_{5})$ corresponds to the state $\ket{h'_5}$, where $V'_5=\{v'_1,v'_2,v'_3,v'_4,v'_5\}$,  $E'_5=\{\{v'_1,v'_2\}, \{v'_1,v'_4\}, \{v'_3,v'_5\}\}$. 
		Each edge of the graph is associated to a controlled-Z gate. Removing the gates $CZ(2,3)$ and $CZ(2,4)$  from the construction of state $\ket{h_5}$, the state $\ket{h'_5}$ can be obtained.}
	\label{fig:partialcontrol}
\end{figure}

In our controlled protocol via $\ket{h_{2N+1}}$, the controller is able to control  each of the $N$ groups to implement operations $\exp[i\b_{j}\s_{n_{O_{j}}}]$ on remote system $O_j$, for $j=N+2,N+3,...,2N+1$. In fact, according to the demand of realistic background,  the controller can control any $k$ of these $N$ groups, and discard the rest of $N-k$ groups with $k\geq1$. It can be realized by employing appropriate graph states. Compared with the construction of $\ket{h_{2N+1}}$ in FIG. \ref{fig:entstate2N+1} and \ref{fig:graph}, such states can be constructed by removing corresponding gates $CZ(2,k)$ and $CZ(k,N+2)$ in pairs, for $k=3,4,...,N+1$.
 We take the protocol via $\ket{h_5}$ as an example. In the former protocol presented in Sec. \ref{sec:fivepartite}, the controller can control  both groups. Now we assume the controller $A_1$ only wants to control  the group consisting of participants $A_2$ and $A_4$. This task can be realized by employing another graph state $\ket{h'_5}$. The construction of $\ket{h'_5}$ can be done with the help of the graph on the right hand side in FIG. \ref{fig:partialcontrol}. That is to say, $\ket{h_5}$ can be constructed by removing the gates $CZ(2,3)$ and $CZ(2,4)$ compared with the construction of state $\ket{h_5}$. From this point of view, our protocol can flexibly follow the actual demands and thus it shows advantage in extensive applications.

\section{geometric measure of entanglement for graph states}
\label{sec:calculateGM}
In this section, we investigate the GM of the graph states used in this paper. Then we compare the entanglement requirement between the protocol in Ref. \cite{reznik2002remote} and our protocol. We find that the two protocols require the same entanglement resource. Hence it is economic to realize the control function from the perspective of entanglement cost.

First we show an example by considering the tripartite graph state $\ket{h_3}$ shown in (\ref{def:stateh3}). Since the local unitary operations do not change the entanglement of the states, we perform the local Hadamard gate $H$ on qubit $a$ of $\ket{h_3}$, and obtain the state  
\begin{eqnarray}
\ket{g_3}=\frac{1}{2}(\ket{000}+\ket{011}+\ket{110}+\ket{101}).
\end{eqnarray} 

Now we consider the GM of $\r_3=\ket{g_3}\bra{g_3}$. According to Lemma \ref{lemma}, its closest product state can be chosen to be non-negative. Let $\ket{\varphi_3}=\otimes_{j=1}^3(\cos\t_j\ket{0}+\sin\t_j\ket{1})$ be a closest product state, with $0\leq \t_j\leq \frac{\pi}{2}$. After some calculations, we obtain that
\begin{eqnarray}
\L^2(\r_3)=&&\max_{\ket{\varphi_3}}\bra{\varphi_3}\rho_3\ket{\varphi_3}=\max_{\ket{\varphi_3}}|\braket{\varphi_3}{g_3}|^2\\
=&&\frac{1}{4}\max_{\t_1,\t_2,\t_3}[\cos\t_1\cos(\t_2-\t_3)+\sin\t_1\sin(\t_2+\t_3)]^2=\frac{1}{2},\nonumber\\
G(\r_3)=&&-2\log\L(\r_3)=1.
\end{eqnarray}
The maximum in the above equation is obtained at $\t_1=\t_2=\t_3=\pi/4$, i.e. $\ket{\varphi_3}=\frac{\sqrt{2}}{4}(\ket{0}+\ket{1})^{\otimes3}$.

Next we consider the GM of ($2N+1$)-partite graph state $\ket{h_{2N+1}}$ in (\ref{def:stateh2n+1}). We show the following fact:
\begin{proposition}
	\label{pro:GM}
	The geometric measure of entanglement for the state $\ket{h_{2N+1}}$ is equal to $N$, for $N\geq2$.
\end{proposition}
\begin{proof}
We reduce the state $\ket{h_{2N+1}}$ by applying the local Hadamard gates  $H$ on qubits $a_3, a_4,...,a_{N+1}$. Then it can be transformed to the following non-negative state,
\begin{eqnarray}
	\ket{g_{2N+1}}=\frac{1}{\sqrt{2^{N+1}}}\sum_{q_2,...,q_{N+1}=0}^1(\!\!\!\!&&\ket{0,q_2,q_3,...,q_{N+1}, q_2,q_3,...,q_{N+1}}\\
	+\!\!\!\!&&\ket{1,q_2\oplus1,q_3\oplus1,...,q_{N+1}\oplus1, q_2,q_3,...,q_{N+1}})_{a_1,a_2,...,a_{2N+1}}.\nonumber
\end{eqnarray} 
So the GM of $\ket{h_{2N+1}}$  is equal to that of $\ket{g_{2N+1}}$. Let $\r_{2N+1}=\ket{g_{2N+1}}\bra{g_{2N+1}}$ and $\ket{\varphi_{2N+1}}=\otimes_{j=1}^{2N+1}(\cos\t_j\ket{0}+\sin\t_j\ket{1})$. We denote $q_k\oplus 1$ as $\overline{q_k}$. Obviously, we have
\begin{eqnarray}
	\label{eq:braketqk}
(\cos\t_k\bra{0}+\sin\t_k\bra{1})\ket{q_k}=\cos^{\overline{q_k}}\t_k\sin^{q_k}\t_k,
\end{eqnarray}
where $q_k=0,1$ and $\cos^{0}\t_k=\sin^{0}\t_k=1$.
Using (\ref{eq:braketqk}), we obtain that
\begin{eqnarray}
\braket{\varphi_{2N+1}}{g_{2N+1}}
=\frac{1}{\sqrt{2^{N+1}}}\sum_{q_2,q_3,...,q_{N+1}}[\!\!\!\!\!&&\cos\t_1(\prod_{k=2}^{N+1}\cos^{\overline{q_k}}\t_k\sin^{q_k}\t_k)(\prod_{l=2}^{N+1}\cos^{\overline{q_{l}}}\t_{l+N}\sin^{q_l}\t_{l+N})\\
&&+\sin\t_1(\prod_{m=2}^{N+1}\cos^{q_m}\t_m\sin^{\overline{q_m}}\t_m)(\prod_{n=2}^{N+1}\cos^{\overline{q_n}}\t_{n+N}\sin^{q_n}\t_{n+N})].\nonumber
\end{eqnarray}
By some calculations, we have
\begin{eqnarray}
	\label{eq:innerproduct}
\braket{\varphi_{2N+1}}{g_{2N+1}}
=\frac{1}{\sqrt{2^{N+1}}}\cos\t_1\prod_{t=2}^{N+1}\cos(\t_t-\t_{t+N})+\sin\t_1\prod_{s=2}^{N+1}\sin(\t_s+\t_{s+N}).
\end{eqnarray}
Using Lemma \ref{lemma}, we have $0\leq\t_j\leq\pi/2$  with $j=1,2,...,2N+1$. From (\ref{eq:innerproduct}), the value of $\braket{\varphi_{2N+1}}{g_{2N+1}}$ reaches the maximum $1/\sqrt{2^N}$ when $\t_j=\pi/4$. Hence we have
\begin{eqnarray}
	\L^2(\r_{2N+1})=\max_{\t_j}|\braket{\varphi_{2N+1}}{g_{2N+1}}|^2=\frac{1}{2^N}.\nonumber
\end{eqnarray}
It holds that
\begin{eqnarray}
G(\r_{2N+1})=-2\log\L(\r_{2N+1})=N.
\end{eqnarray}
\end{proof}

Now we compare the entanglement requirement between the protocol in Ref. \cite{reznik2002remote} and our controlled protocol in terms of GM. The protocol for remote operations by stator is firstly proposed in Ref. \cite{reznik2002remote}. In this protocol,  $2N$ participants share the   $2N$-partite entangled state $\ket{\phi_{2N}}$ to realize the  remote operations on $N$ systems without the controller. The entangled states used in that protocol are given as follows,
\begin{eqnarray}
	\label{def:statephi2}
	\ket{\phi_{2}}=&&\frac{1}{\sqrt{2}}(\ket{00}+\ket{11}),\\
		\label{def:statephi4}
	\ket{\phi_{4}}=&&\frac{1}{2}(\ket{0000}+\ket{0101}+\ket{1010}+\ket{1111}),\\
	\vdots\nonumber\\
		\label{def:statephi2n+1}
\ket{\phi_{2N}}=&&\frac{1}{\sqrt{2^N}}\sum_{q_1,q_2,...,q_{N}=0}^1\ket{q_1,q_2,...,q_N,q_1,q_2,...,q_N}.
\end{eqnarray}

 We investigate GM of the states in (\ref{def:statephi2})-(\ref{def:statephi2n+1}) and list the entangled states required in the two protocols for a given number of systems.  The results are presented in TABLE. \ref{table:GM}.

\begin{table}[h]
	\caption{The comparison of entanglement resource (entangled states and its GM) between our controlled implementation protocol  and the protocol in Ref. \cite{reznik2002remote} for a given number of remote system $N$. }
	\label{table:GM}
	\begin{tabular}{|c|c|c|c|c|}	
		\hline
	&   \multicolumn{2}{c|}{\makecell[c]{Controlled remote implementation\\ protocol in this paper}}   &\multicolumn{2}{c|}{\makecell[c]{Remote implementation protocol\\ in Ref. \cite{reznik2002remote}}}   \\
		\hline
Number of systems&Entangled state&GM&Entangled state&GM\\
\hline
1 & $\ket{h_3}$ in (\ref{def:stateh3})  &1  &$\ket{\phi_{2}}$ in (\ref{def:statephi2}) &1 \\
\hline
2 & $\ket{h_5}$ in (\ref{def:stateh5})  &2  &$\ket{\phi_{4}}$ in (\ref{def:statephi4}) &2 \\
\hline
\vdots & \vdots  &\vdots  &\vdots &\vdots \\
\hline
for any integer $N$ & $\ket{h_{2N+1}}$ in (\ref{def:stateh2n+1})  &$N$  &$\ket{\phi_{2N}}$ in (\ref{def:statephi2n+1}) &$N$ \\
\hline
	\end{tabular}
\end{table}
From TABLE. \ref{table:GM}, we see that the protocol we proposed in this paper requires the same entanglement resource as the former protocol in Ref. \cite{reznik2002remote}. That is to say, although one "controller" who enables control over other groups is added in our protocol, the entanglement cost of our protocol does not increase. Hence it is economic to realize the control function from the perspective of entanglement cost.

\section{control power analysis}
\label{sec:POVM}
In this section, we show the control power analysis of our controlled protocol. In analogy to the control power of quantum teleportation \cite{li2015analysis}, the controller's power here is determined by the situation that the remote operation be accomplished without the controller's help. If the controller does not wish the remote operation to be executed, he will not perform the measurement and other participants can hardly realize the target operations on remote systems.
We show that without the controller's permission, eight specified operations can be implemented with the success rate of 50\%, other operations can be realized with the success rate of 25\%. It means that the control power is reliable in our controlled remote implementation protocol.  The main result of this section is shown in Proposition \ref{pro:controlpower}.

We consider the protocol via $\ket{h_{3}}$ in Sec. \ref{sec:tripartite}. Other protocols can be analyzed in the same way, as the result we obtained in this section is valid for each group  in the protocol via $\ket{h_{2N+1}}$. To implement the remote operation $e^{i\a\s_{n_C}}$ without the controller Alice's help, Bob and Charlie collaborate to perform some POVM measurement on qubits $b$ and $c$ to eliminate the entanglement between qubit $a$ and $b,c$.  In particular, Bob carries out the POVM  $\{M_j=\ket{\b_j}\bra{\b_j}\}$ and Charlie carries out $\{N_k=\ket{\g_k}\bra{\g_k}\}$, for $j,k=1,2$. The operators satisfy  two basic restrictions,
\begin{eqnarray}
	\label{def:conditionbj}
\sum_{j=1}^2M_j^\dg M_j=\sum_{j=1}^2\ket{\b_{j}}\bra{\b_{j}}=I,\\
\label{def:conditiongk} 
\sum_{k=1}^2N_k^\dg N_k=\sum_{k=1}^2\ket{\g_{k}}\bra{\g_{k}}=I.
\end{eqnarray}

We investigate the probability by which POVM operators $M_j$ and $N_k$ occur in the measurement. It is denoted as $p(j,k)$. We normalize the stator  $S_3'$ in (\ref{def:statorS3'}) and define the normalized stator as $W_3$, where
\begin{eqnarray}
	W_3=\frac{S_3'}{\sqrt{\tr(S_3'^\dg S_3')}}=\frac{1}{\sqrt{2}}S_3'.
\end{eqnarray}
Note that $\tr_C(\s_{n_C})=0$. So we  have
\begin{eqnarray}
	\label{def:probability}
	p(j,k)=&&\tr_C\left[W_3^\dg(I_a\otimes M_{j_b}^\dg M_{j_b} \otimes N_{k_c}^\dg N_{k_c})W_3\right]\nonumber\\
	=&&\frac{1}{8}\tr_C\left[I_C+\sin(2\t_j)\sin(2\l_k)\cos\varphi_j\cos\og_k\s_{n_C}\right]\\
	=&&\frac{1}{8}\tr_C(I_C)=\frac{1}{4}.\nonumber
\end{eqnarray} 
That is to say, the probability of performing $ M_{j_b}^\dg M_{j_b} \otimes N_{k_c}^\dg N_{k_c}=M_{j_b}\otimes N_{k_c}$ on qubits $b$ and $c$ is equal to 25\% for any $(j,k)=(1,1),...(2,2)$. 

Suppose $\ket{\b_j}=\cos\t_j\ket{0}+e^{i\varphi_j}\sin\t_j\ket{1}$ and $\ket{\g_k}=\cos\l_k\ket{0}+e^{i\og_k}\sin\l_k\ket{1}$. Without loss of generality, we assume $\t_j, \l_k\in[0,\frac{\pi}{2}]$ and $\varphi_j,\og_k\in [0,2\pi)$. 
Let  $c_{s,t}=\braket{\b_j}{s}\braket{\g_k}{t}$ with $s,t=0,1$, i.e.
\begin{eqnarray}
	\label{def:c00}
	c_{0,0}=&&\cos\t_j\cos\l_k,\\
	c_{0,1}=&&e^{-i\og_k}\cos\t_j\sin\l_k,\\
	c_{1,0}=&&e^{-i\varphi_j}\sin\t_j\cos\l_k,\\
	\label{def:c11}
	c_{1,1}=&&e^{-i(\og_k+\varphi_j)}\sin\t_j\sin\l_k.
\end{eqnarray}

We consider the situation that POVM operators $M_j$ and $N_k$ occur in the measurement, for $j,k=1,2$. After the measurement, the stator $S_3'$ in (\ref{def:statorS3'}) becomes
\begin{eqnarray}
	\label{def:T31}
	T_3'=&&\frac{1}{2}\ket{+}_a(M_j\ket{0}_{b}N_k\ket{0}_c\otimes I_C+M_j\ket{1}_{b}N_k\ket{1}_c\otimes \s_{n_C})\nonumber\\
	+&&\frac{1}{2}\ket{-}_a(M_j\ket{1}_{b}N_k\ket{0}_c\otimes I_C+M_j\ket{0}_{b}N_k\ket{1}_c\otimes \s_{n_C})\\
	=&&\frac{1}{2}\ket{+}_a(c_{0,0}\ket{\b_{j},\g_k}_{b,c}\otimes I_C+c_{1,1}\ket{\b_{j},\g_k}_{b,c}\otimes \s_{n_C})\nonumber\\
	+&&\frac{1}{2}\ket{-}_a(c_{1,0}\ket{\b_{j},\g_k}_{b,c}\otimes I_C+c_{0,1}\ket{\b_{j},\g_k}_{b,c}\otimes \s_{n_C}).\nonumber
\end{eqnarray}

Bob and  Charlie aim to implement the operation 
\begin{eqnarray}
	\label{def:operation}
e^{i\a\s_{n_C}}=\cos\a I_C+i\sin\a\s_{n_C}
\end{eqnarray}
on system $C$ without Alice's help. Their goal is to make qubit $a$ separate from remaining qubits $b$ and $c$ by POVM measurement. 
We analyze the possible angle $\a$ in the operation $e^{i\a\s_{n_C}}$ that can be realized in this scheme, and obtain the following fact:
\begin{proposition}
	\label{pro:controlpower}
For the protocol via $\ket{h_{3}}$ in Sec. \ref{sec:tripartite}, the probabilistic implementation of remote operation $e^{i\a\s_{n_C}}$ can be realized by POVM measurement without controller's permission. The success rate to realize the operation $e^{i\a\s_{n_C}}$ is equal to 50\% for $\a\in\left\{0,\frac{\pi}{4},\frac{\pi}{2},\frac{3\pi}{4},\pi,\frac{5\pi}{4},\frac{3\pi}{2},\frac{7\pi}{4}\right\}$, and that is equal to 25\% for $\a\in\bigcup_{m=0}^7(\frac{m\pi}{4},\frac{(m+1)\pi}{4})$.
\end{proposition}

We prove  Proposition \ref{pro:controlpower} in Appendix \ref{appendix:controlpower}.  

Now we show some security analysis of the POVM measurement scheme. In our controlled protocol in Sec. \ref{sec:tripartite}, the deterministic implementation of $e^{i\a\s_{n_C}}$ can be realized with the permission of controller Alice. The rotation angle $\a$  can be chosen as any value and Charlie is unaware of any information of the operation that Bob wants to implement. However, in this POVM measurement scheme in Sec. \ref{sec:POVM}, we claim that Charlie can obtain the value of $\a$ with the probability of $25\%$ or 12.5\%. It means that the confidentiality of remote operation $e^{i\a\s_{n_C}}$ may be destroyed. In fact, the rotation angle $\a$ in $e^{i\a\s_{n_C}}$ can only be one of \{$\frac{n\pi}{2}$ with $n=0,1,2,3$\}  and \{$\frac{m\pi}{4}$ with $m=1,3,5,7$\} when Charlie prepares the POVM operators with $\l_k=0$ and $\l_k=\frac{\pi}{4}$, respectively. So there are four possible rotation angles $\a$ for a group of POVM operators with given parameters.  Charlie may guess what the rotation angle $\a$ is, and get to the correct answer with the probability of 25\%.  When Charlie prepares the operator with $\l_1\in(0,\frac{\pi}{4})\cup(\frac{\pi}{4},\frac{\pi}{2})$, the target rotation angle can only be one of $\{\pi-\l_1, 2\pi-\l_1,  \frac{3\pi}{2}-\l_1, \frac{\pi}{2}-\l_1, \l_1,\l_1+\pi, \l_1+\frac{3\pi}{2}, \l_1+\frac{\pi}{2}\}$, so Charlie may be aware of the value of $\a$ with the probability of 12.5\%. 
	
To sum up, Bob and Charlie can hardly realize the remote implementation of operation perfectly without the permission of controller Alice. The control power in our  controlled protocol is convincing.

\section{experimental feasibility}
\label{sec:experimental}
As the entanglement resource of our protocol, graph states are the most readily available multipartite resource in the laboratory, and they have already been built and used for information processing experimentally. In a recent work,  the deterministic protocol is implemented, by which the GHZ states of up to 14 photons and linear cluster states of up to 12 photons have been grown with a fidelity lower bounded by 76(6)\% and 56(4)\%, respectively \cite{thomas2022efficient}. A scheme to prepare an ultrahigh-fidelity four-photon linear cluster state has been proposed, and it has been experimentally demonstrated with the fidelity of 0.9517 $\pm$ 0.0027.  This scheme can be directly extended to more photons to generate an N-qubit linear cluster state \cite{zhang2016experimental}.  Recently, an efficient scheme is demonstrated to prepare graph states with only a polynomial overhead using long-lived atomic quantum memories \cite{zhang2022quantum}. Such technique can be used to prepare the graph states (\ref{def:stateh2n+1}) in our protocol, so the requirement of entanglement resource can be satisfied. 

Remote implementation has been demonstrated theoretically and realized experimentally,  such as  remote state preparation \cite{pogorzalek2019secure}, nonlocal CNOT operation implementation \cite{zhou2019hyper}  and remote generation of entanglement \cite{morin2014remote}. In our controlled remote implementation protocol, the participants prepare the stators and implement the remote operation by applying LOCC on shared entangled state. The local operations and measurements required in our protocol are $U_{a_j,O_j}, H, \s_{x}, \s_{z}$ operations and X-basis, Z- basis measurements, respectively. In the work presented in \cite{bhaskar2020experimental}, a kind of memory-enhanced quantum communication is demonstrated experimentally. In this work, some local operations and measurements in the X and Z bases are implemented with a time-delay interferometer (TDI). The device consists of a diamond nanophotonic resonator containing SiV quantum memory  with an integrated microwave stripline.  By reducing the possibility that an additional photon reaches the cavity, high spin-photon gate fidelities are enabled. Measurements are performed with high-fidelity by keeping track of the timing when the TDI piezo voltage reaches a limiting value, which guarantees that the SiV is always resonant with the photonic qubits.  The experimental technique in this work can be employed to realize our protocol. Hence our protocol is feasible according to the present experimental technologies.  Additionally, the memory-based communication nodes allow asynchronous Bell-state measurement and it may enhance the performance of our protocol experimentally.

\section{conclusion}
\label{sec:conclusion}
We have proposed the protocol for controlled remote implementation of operations in the form of $U=\exp[i\a\s_{n}]$ for each system. A family of graph states is constructed as the entanglement resource in our protocol. Sharing the $(2N+1)$-partite graph state as channel, $2N$ participants are able to realize the remote operations $U$ on $N$ unknown states for distant systems respectively, only with the permission of a controller. The implementation requirements of our protocol can be satisfied only by means of LOCC. Further we have characterized GM of the graph states in our protocol. Compared with the entanglement requirement of the protocol in \cite{reznik2002remote}, the control function of our protocol is realized economically. Based on the result of control power analysis, the control power is reliable  in our protocol, i.e. others can hardly realize the implementation without the permission of controller. Further we have exhibited the experimental feasibility of our protocol in terms of current techniques.

Many problems arising from this paper can be further explored. As we all know, multilevel systems as qudits feature more advantages than their binary counterpart.  Our protocol can be considered in high-dimensional Hilbert spaces, and may be used to implement some other operations. The protocol with more controllers for specific systems can be studied by applying appropriate graph states and local implementations. The interaction between quantum states and environment is unavoidable, leading to the loss of accuracy. Considering the situation that graph states in our protocol are affected by noise, the probabilistic implementation of that can be developed. 

\section*{ACKNOWLEDGMENTS}
We thank Li Yu,  Jun Li and Zhaohui Wei for careful reading of the whole paper. The authors were supported by the NNSF of China (Grant No. 11871089) and the Fundamental Research Funds for the Central Universities (Grant No. ZG216S2005).

\section*{AUTHOR CONTRIBUTIONS}
 All authors contributed to the discussion of results and writing of the manuscript.

\section*{COMPETING INTERESTS}
The authors declare no competing interests.

  \appendix
 \section{control power analysis in terms of POVM measurement}
 \label{appendix:controlpower}
 In this section, we show the proof of proposition \ref{pro:controlpower} in Sec. \ref{sec:POVM}.
 
\begin{proof}
	The proof includes two parts: first we consider case $\textrm{I}$ for $c_{0,0}=c_{1,0}=0$ or $c_{0,1}=c_{1,1}=0$ in TABLE \ref{table:coefficients} in Sec. \ref{sec:case1}, then we consider case $\textrm{II}$ for $c_{0,1}c_{1,0}\neq0$ in TABLE \ref{table:coefficients2} in Sec. \ref{sec:case2}, where $c_{s,t}$ with $s,t=1,2$ are defined in (\ref{def:c00})-(\ref{def:c11}).
	
	For other cases, the entanglement may not be eliminated or the same results are obtained as that in the former two cases. Case $\textrm{III}$ for $c_{0,0}c_{1,1}\neq0$ can be analyzed by the same way as case $\textrm{II}$ and the same results are obtained, as the two cases both require two sets of coefficients $\{c_{0,0},c_{1,1}\}$ and $\{c_{1,0},c_{0,1}\}$ to be proportional, i.e. $\frac{c_{0,0}}{c_{1,0}}=\frac{c_{1,1}}{c_{0,1}}$ in case $\textrm{II}$ and $\frac{c_{1,0}}{c_{0,0}}=\frac{c_{0,1}}{c_{1,1}}$ in case  $\textrm{III}$. In case $\textrm{IV}$ for $c_{0,0}c_{0,1}\neq0$ and case $\textrm{V}$ for $c_{1,0}c_{1,1}\neq0$, the entanglement can be eliminated only when $c_{s,t}\neq0$ with $s,t=1,2$, which is included in case $\textrm{II}$.   
	
	\subsection{ Case \textrm{I} : $c_{0,0}=c_{1,0}=0$ or $c_{0,1}=c_{1,1}=0$ }
	\label{sec:case1}
	The first case that $c_{0,0}=c_{1,0}=0$ or $c_{0,1}=c_{1,1}=0$  corresponds to $\l_k=\frac{\pi}{2}$ or $\l_k=0$ in (\ref{def:c00})-(\ref{def:c11}), respectively.  In this case, the operations $e^{i\a\s_{n_C}}$ with $\a=0,\frac{\pi}{2},\pi, \frac{3\pi}{2}$ can be realized by TABLE \ref{table:coefficients}, and the success rate of realization is equal to  50\%.
	
	As an example, the operations $e^{i\a\s_{n_C}}$ with $\a=0$ or $\pi$ and $\a=\frac{\pi}{2}$ or $\frac{3\pi}{2}$ can be realized by different combinations of POVM operators. Bob and Charlie prepare the POVM operators with the parameters $\t_1=0, \t_2=\frac{\pi}{2}, \varphi_1=\varphi_2=0$, and $\l_1=0,\l_2=\frac{\pi}{2}, \og_1=\og_2=\frac{\pi}{2}$. 
	Their POVM operators are
	\begin{eqnarray}
		M_1=N_1=\bma 1& 0\\ 0& 0\ema,\quad
		M_2=N_2=\bma 0& 0\\ 0&1\ema.
	\end{eqnarray}
	When the POVM operators $(M_j,N_k)$ occurs in the measurement, the stator $T_3'$ in (\ref{def:T31}) becomes $H_3^{(j,k)}$, where
	\begin{eqnarray}
		H_3^{(1,1)}=&&\frac{1}{2}\ket{+}_a\ket{\b_{1},\g_1}_{b,c}\otimes I_C,\\
		H_3^{(1,2)}=&&\frac{1}{2}\ket{-}_a\ket{\b_{1},\g_2}_{b,c}\otimes (-i\s_{n_C}),\\
		H_3^{(2,1)}=&&\frac{1}{2}\ket{-}_a\ket{\b_{2},\g_1}_{b,c}\otimes I_C,\\
		H_3^{(2,2)}=&&\frac{1}{2}\ket{+}_a\ket{\b_{2},\g_2}_{b,c}\otimes(-i\s_{n_C}).
	\end{eqnarray}
	Note that $I_C\propto e^{i0\s_{n_C}}\propto e^{i\pi\s_{n_C}}$,  and $-i\s_{n_C}\propto e^{i\frac{\pi}{2}\s_{n_C}}\propto e^{i\frac{3\pi}{2}\s_{n_C}}$. Hence, the operation $e^{i\a\s_{n_C}}$ with rotation angle $\a=0$ or $\pi$ can be realized when $(M_1,N_1)$ and  $(M_2,N_1)$ occur in the measurement;  the operation with rotation angle $\a=\frac{\pi}{2}$ or $\frac{3\pi}{2}$ can be realized when $(M_1,N_2)$ and  $(M_2,N_2)$ occur in the measurement. From (\ref{def:probability}), each of the four combinations of POVM operators  occurs with the probability of $25\%$.  So Bob can only implement his desired operation with the success rate of 50\%. For example, Bob wishes to implement the operation $e^{i\pi\s_{n_C}}$ on system $C$, he can only carry it off when the POVM operators $(M_1,N_1)$ and $(M_2,N_1)$ occur, and the probability of that is 50\%.  
	
	The operations $e^{i\a\s_{n_C}}$ with $\a=0$ or $\pi$ and $\a=\frac{\pi}{2}$ or $\frac{3\pi}{2}$ can also be realized by other POVM operators with different parameters. The process of implementation is the same as that of the former example. Combined with (\ref{def:conditionbj}) and (\ref{def:c00})-(\ref{def:c11}), we obtain the possible parameters in Bob and Charlie's POVM operators and its corresponding rotation angle $\alpha$ in $e^{i\a\s_{n_C}}$  that can be realized. They are shown in TABLE \ref{table:coefficients}. Each of the four combinations of POVM operators $(M_j,N_k)$,  occurs with the same probability of $25\%$, for $j,k=1,2$.  In this case, Bob can realize his desired operation with the probability of 50\%, and  his target operation can only be  $e^{i\a\s_{n_C}}$ with the rotation angle $\a=\frac{n\pi}{2}$, for $n=0,1,2,3$. 
	
	\begin{table}[h]
		\renewcommand\arraystretch{1.25}
		\caption{The possible parameters $\t_j, \varphi_j, \l_k, \og_k$ in Bob and Charlie's POVM operators and its corresponding rotation angle $\alpha$ that can be realized, for $j,k=1,2$.   Column 9 contains the choice of POVM operators $(M_j, N_k)$, which is associated to different rotation angle $\a$. Columns 10-13 contain the coefficients in (\ref{def:T31})  under specific POVM operators.}
		\label{table:coefficients}
		\begin{tabular}{|c|c|c|c|c|c|c|c|c|c|c|c|c|c|}
			\hline
			\multicolumn{4}{|c|}{\makecell[c]{parameters in Bob's\\POVM operator}}& \multicolumn{4}{c|}{\makecell[c]{parameters in Charlie's\\POVM operator}}&\makecell[c]{POVM\\ operators} &   \multicolumn{4}{c|}{coefficients in (\ref{def:T31}) }&\makecell[c]{rotation angle\\ in $e^{i\a\s_{n_C}}$}     \\
			\hline
			$\t_1$&$\t_2$&$\varphi_1$&$\varphi_2$&$\l_1$&$\l_2$&$\og_1$&$\og_2$&$M_j,N_k$ &$c_{0,0}$&$c_{0,1}$&$c_{1,0}$&$c_{1,1}$&$\a$  \\
			\hline
			\multirow{4}{*}{[0,$\frac{\pi}{2}$] }   & \multirow{4}{*}{$\frac{\pi}{2}-\t_1$}& \multirow{4}{*}{$[0,2\pi]$} &\multirow{4}{*}{ $\varphi_1\pm\pi$} &\multirow{4}{*}{0}   &\multirow{4}{*}{$\frac{\pi}{2}$} &\multirow{4}{*}{ $[0,2\pi]$}&  \multirow{4}{*}{ $[0,2\pi]$}&$M_1,N_1$&$\cos\t_1$&0&$e^{-i\varphi_1}\sin\t_1$&0&   0 or $\pi$    \\
			\cline{9-14}
			&   & & &  &  &  &  &$M_1,N_2$&0&$e^{-i\og_2}\cos\t_1$&0&$e^{-i(\og_2+\varphi_1)}\sin\t_1$&   $\frac{\pi}{2}$ or $\frac{3\pi}{2}$  \\
			\cline{9-14}
			&   & & &  &  &  &  &$M_2,N_1$&$\cos\t_2$&0&$e^{-i\varphi_2}\sin\t_2$&0& 0 or $\pi$  \\
			\cline{9-14}
			&   & & &  &  &  &  &$M_2,N_2$&0&$e^{-i\og_2}\cos\t_2$&0&$e^{-i(\og_2+\varphi_2)}\sin\t_2$&   $\frac{\pi}{2}$ or $\frac{3\pi}{2}$  \\
			\hline
			\multirow{4}{*}{[0,$\frac{\pi}{2}$] }   & \multirow{4}{*}{$\frac{\pi}{2}-\t_1$}&\multirow{4}{*}{$[0,2\pi]$} &\multirow{4}{*}{ $\varphi_1\pm\pi$} &\multirow{4}{*}{$\frac{\pi}{2}$}   &\multirow{4}{*}{0} &\multirow{4}{*}{ $[0,2\pi]$}&  \multirow{4}{*}{ $[0,2\pi]$}&$M_1,N_1$&0&$e^{-i\og_1}\cos\t_1$&0&$e^{-i(\og_1+\varphi_1)}\sin\t_1$&    $\frac{\pi}{2}$ or $\frac{3\pi}{2}$   \\
			\cline{9-14}
			&   & & &  &  &  &  &$M_1,N_2$&$\cos\t_1$&0&$e^{-i\varphi_1}\sin\t_1$&0& 0 or $\pi$  \\
			\cline{9-14}
			&   & & &  &  &  &  &$M_2,N_1$&0&$e^{-i\og_1}\cos\t_2$&0&$e^{-i(\og_1+\varphi_2)}\sin\t_2$& $\frac{\pi}{2}$ or $\frac{3\pi}{2}$ \\
			\cline{9-14}
			&   & & &  &  &  &  &$M_2,N_2$&$\cos\t_2$&0&$e^{-i\varphi_2}\sin\t_2$&0& 0 or $\pi$  \\
			\hline
		\end{tabular}
	\end{table}
	
	\subsection{Case \textrm{II}: $c_{0,1}c_{1,0}\neq0$}
	\label{sec:case2}
	Next we consider the second case that $c_{0,1}c_{1,0}\neq0$, i.e. $\t_j,\l_k\in(0, \frac{\pi}{2})$ in (\ref{def:c00})-(\ref{def:c11}). We set $\l_k\neq0,\frac{\pi}{2}$ and  $\a\neq\frac{n\pi}{2}$ with   $n=0,1,2,3$ here, as they have been discussed in the former case. To eliminate the entanglement between qubit $a$ and $b,c$, the coefficients $c_{s,t}$ in (\ref{def:T31}) should satisfy the following restriction,
	\begin{eqnarray}
		\label{eq:cK}
		\frac{c_{0,0}}{c_{1,0}}=\frac{c_{1,1}}{c_{0,1}}=K,
	\end{eqnarray}
	where $K\in\bbC$ is a constant. So the stator in (\ref{def:T31}) becomes
	\begin{eqnarray}
		\label{def:T32}
		T_3'=\frac{1}{2}(K\ket{+}_a+\ket{-}_a)\ket{\b_{j},\g_k}_{b,c}\otimes(c_{1,0}I_C+c_{0,1}\s_{n_C}).
	\end{eqnarray}
	The target operation that Bob and  Charlie try to implement on system $C$ is  $e^{i\a\s_{n_C}}=\cos\a I_C+i\sin\a\s_{n_C}$. So the coefficients in (\ref{def:T32}) should satisfy
	\begin{eqnarray}
		\label{eq:c10c01X}
		c_{1,0}=L\cos\a,\quad c_{0,1}=iL\sin\a,
	\end{eqnarray}
	where $L=x+iy\in\bbC$ is a constant with $x\neq0$ or $y\neq0$, $x,y\in\bbR$. 
	From (\ref{def:c00})-(\ref{def:c11}), (\ref{eq:cK}) and (\ref{eq:c10c01X}),  we obtain that
	\begin{eqnarray}
		\label{eq:c00}
		c_{0,0}=&&\cos\t_j\cos\l_k=KL\cos\a,\\
		\label{eq:c01}
		c_{0,1}=&&e^{-i\og_k}\cos\t_j\sin\l_k=iL\sin\a,\\
		\label{eq:c10}
		c_{1,0}=&&e^{-i\varphi_j}\sin\t_j\cos\l_k=L\cos\a,\\
		\label{eq:c11}
		c_{1,1}=&&e^{-i(\og_k+\varphi_j)}\sin\t_j\sin\l_k=iKL\sin\a.
	\end{eqnarray}
	
	Now we analyze the possible angle $\a$ that satisfies the conditions above. Note that  $\t_j,\l_k\in(0,\frac{\pi}{2})$. From (\ref{eq:c00})-(\ref{eq:c11}), we have 
	\begin{eqnarray}
		e^{2i\varphi_j}\cos^2\t_j=\sin^2\t_j.
	\end{eqnarray}
	So we have $e^{2i\varphi_j}\in \bbR$ and thus
	\begin{eqnarray}
		\label{def:varphij}
		&&\varphi_j=0\;\mbox{or} \; \pi,\\
		\label{def:tj}
		&&\t_1=\t_2=\frac{\pi}{4}.
	\end{eqnarray}

	Then from (\ref{eq:c10}) and (\ref{def:varphij}), we have
	\begin{eqnarray}
		\sin\t_j\cos\l_k=x\cos\a+iy\cos\a,
	\end{eqnarray}
	or
	\begin{eqnarray}
		-\sin\t_j\cos\l_k=x\cos\a+iy\cos\a.
	\end{eqnarray}
	Obviously, $x\cos\a+iy\cos\a\in\bbR$. Note that $\t_j,\l_k\in(0,\frac{\pi}{2}) $ and thus $\cos\a\neq 0$.
	Hence $y=0$, i.e. $L=x\in \bbR$. From (\ref{eq:c01}), we have
	\begin{eqnarray}
		\label{def:ogk}
		\og_k=\frac{\pi}{2}\;\mbox{or}\;\frac{3\pi}{2}.
	\end{eqnarray}
	Using (\ref{eq:c01}),  (\ref{eq:c10}) and (\ref{def:tj}), we have 
	\begin{eqnarray}
		\label{eq:X}
		\frac{1}{2}(\cos^2\l_k+\sin^2\l_k)=L^2(\cos^2\a+\sin^2\a)\Rightarrow L=\pm\frac{\sqrt{2}}{2}.
	\end{eqnarray}
	From (\ref{def:conditionbj}), (\ref{def:conditiongk}), (\ref{def:varphij}), (\ref{def:tj}) and (\ref{def:ogk}), we obtain that
	\begin{eqnarray}
		\label{eq:t1og1}
		&&\varphi_1\neq\varphi_2, \; \og_1\neq\og_2,\\
		\label{eq:l1l2}
		&&\l_1+\l_2=\frac{\pi}{2}.
	\end{eqnarray}
	Considering the restrictions (\ref{def:varphij}), (\ref{def:tj}), (\ref{def:ogk})-(\ref{eq:l1l2}) on Eqs. (\ref{eq:c01}) and (\ref{eq:c10}), we show that the remote operation $e^{i\a\s_{n_C}}$ with the rotation angle 
	\begin{eqnarray}
		\label{def:alpharange}
		\a\in\bigcup_{m=0}^7(\frac{m\pi}{4},\frac{(m+1)\pi}{4})
	\end{eqnarray} 
	can be realized with the success rate of 25\%, and the operations corresponding to
	\begin{eqnarray}
		\label{eq:alpharestriction}
		\a=\frac{n\pi}{4} \quad \mbox{with}\quad  n=1,3,5,7.
	\end{eqnarray}
	can be realized with the success rate of 50\%. The implementation of these operations $e^{i\a\s_{n_C}}$ with $\a\in \cup_{m=0}^3(\frac{m\pi}{2},\frac{(m+1)\pi}{2})$ can be realized by TABLE \ref{table:coefficients2}.

	As an example, Bob wants to realize the operation $e^{i\a\s_{n_C}}$ with $\a=\frac{3\pi}{4}$. For this purpose,  Bob and Charlie choose the POVM operators with the parameters $\t_1=\t_2=\frac{\pi}{4}$, $\varphi_1=0,\varphi_2=\pi$ and  $\l_1=\l_2=\frac{\pi}{4}$, $\og_1=\frac{\pi}{2},\og_2=\frac{3\pi}{2}$, respectively. Their POVM operators are
	\begin{eqnarray}
		M_1=\ket{\b_1}\bra{\b_1}=\bma \frac{1}{2}& \frac{1}{2}\\ \frac{1}{2}& \frac{1}{2}\ema,\quad
		M_2=\ket{\b_2}\bra{\b_2}=\bma \frac{1}{2}& -\frac{1}{2}\\ -\frac{1}{2}& \frac{1}{2}\ema,\\
		N_1=\ket{\g_1}\bra{\g_1}=\bma\frac{1}{2}& -\frac{1}{2}i\\ \frac{1}{2}i& \frac{1}{2}\ema,\quad
		N_2=\ket{\g_2}\bra{\g_2}=\bma \frac{1}{2}& \frac{1}{2}i\\ -\frac{1}{2}i& \frac{1}{2}\ema.
	\end{eqnarray}
	When the POVM operators $(M_j,N_k)$ occurs in the measurement, the stator $T_3'$ in (\ref{def:T32}) becomes $T_3^{(j,k)}$, where
	\begin{eqnarray}
		\label{def:T3(1,1)}
		T_3^{(1,1)}=&&\frac{1}{2}(\ket{+}_a+\ket{-}_a)\ket{\b_{1},\g_1}_{b,c}\otimes(\frac{1}{2}I_C-\frac{1}{2}i\s_{n_C}),\\
		\label{def:T3(1,2)}
		T_3^{(1,2)}=&&\frac{1}{2}(\ket{+}_a+\ket{-}_a)\ket{\b_{1},\g_2}_{b,c}\otimes(\frac{1}{2}I_C+\frac{1}{2}i\s_{n_C}),\\
		\label{def:T3(2,1)}
		T_3^{(2,1)}=&&\frac{1}{2}(-\ket{+}_a+\ket{-}_a)\ket{\b_{2},\g_1}_{b,c}\otimes(-\frac{1}{2}I_C-\frac{1}{2}i\s_{n_C}),\\
		\label{def:T3(2,2)}
		T_3^{(2,2)}=&&\frac{1}{2}(-\ket{+}_a+\ket{-}_a)\ket{\b_{2},\g_2}_{b,c}\otimes(-\frac{1}{2}I_C+\frac{1}{2}i\s_{n_C}).
	\end{eqnarray}
	Note that $(\frac{1}{2}I_C-\frac{1}{2}i\s_{n_C})\propto e^{i\frac{3\pi}{4}\s_{n_C}}\propto e^{i\frac{7\pi}{4}\s_{n_C}}$ and $(\frac{1}{2}I_C+\frac{1}{2}i\s_{n_C})\propto e^{i\frac{\pi}{4}\s_{n_C}}\propto e^{i\frac{5\pi}{4}\s_{n_C}}$.
	When the POVM operators $(M_1,N_1)$, $(M_1,N_2)$, $(M_2,N_1)$, $(M_2,N_2)$ occur, the operation $e^{i\a_t\s_{n_C}}$ with $t=1,2,3,4$ can be realized respectively, where $\a_1=\a_4=\frac{3\pi}{4}$ or $\frac{7\pi}{4}$, $\a_2=\a_3=\frac{\pi}{4}$ or $\frac{5\pi}{4}$.
	Bob's goal is to implement the operation $e^{i\a\s_{n_C}}$ with  $\a=\frac{3\pi}{4}$. It can only be realized when $(M_1,N_1)$ or $(M_2,N_2)$ occur. From (\ref{def:probability}), each of the four combinations of POVM operators  $(M_j,N_k)$ occurs with the same probability of $25\%$, for $j,k=1,2$. So the success rate of realizing Bob's target operation is 50\%.

	Other possible POVM operators with different parameters and its corresponding rotation angle $\a$ in the operation $e^{i\a\s_{n_C}}$   are listed in TABLE \ref{table:coefficients2}. For the first set of POVM operators with parameter $\l_1=\frac{\pi}{4}$, Bob can realize his desired operation with the probability of 50\% and the target operation that Bob can choose is $e^{i\a\s_{n_C}}$ with $\a=\frac{n\pi}{4}$ with  $n=1,3,5,7$. By the second set of operators with $\l_1\in(0,\frac{\pi}{4})\cup(\frac{\pi}{4},\frac{\pi}{2})$, the operation $e^{i\a\s_{n_C}}$ with $\a\in\cup_{m=0}^7(\frac{m\pi}{4},\frac{(m+1)\pi}{4})$ can be realized. If Bob wants to implement that operation, he informs Charlie to prepare the POVM operators with four parameters: $\l_1$, $\l_2=\frac{\pi}{2}-\l_1$, $\og_1=\frac{\pi}{2},\og_2=\frac{3\pi}{2}$, where 
	\begin{eqnarray}
		\l_1=\begin{cases}
			\a &\mbox{for $\a\in(0,\frac{\pi}{2})$,}\\
			\pi-\a &\mbox{for $\a\in(\frac{\pi}{2},\pi)$,}\\
			\frac{3\pi}{2}-\a &\mbox{for $\a\in(\pi,\frac{3\pi}{2})$,}\\
			\a-\frac{3\pi}{2} &\mbox{for $\a\in(\frac{3\pi}{2},2\pi)$.}\\
		\end{cases}
	\end{eqnarray}
	\begin{table}[h]
	\renewcommand\arraystretch{1.25}
	\caption{The possible parameters $\t_j, \varphi_j, \l_k, \og_k$ in Bob and Charlie's POVM operators and its corresponding rotation angle $\alpha$ in operations $e^{i\a\s_{n_C}}$ that can be realized.  The first column contains the choice of POVM operators $(M_k, N_k)$, for $k=1,2$.}
	\label{table:coefficients2}
	\begin{tabular}{|p{0.8cm}<{\centering}|p{0.8cm}<{\centering}|p{0.8cm}<{\centering}|p{0.8cm}<{\centering}|p{1cm}<{\centering}|p{1cm}<{\centering}|p{0.8cm}<{\centering}|p{0.8cm}<{\centering}|c|c|c|c|c|c|}
		\hline
		\multicolumn{4}{|c|}{\makecell[c]{parameters in Bob's\\POVM operators $M_1,M_2$ }}& \multicolumn{4}{c|}{\makecell[c]{parameters in Charlie's\\POVM operators $N_1,N_2$}}&\makecell[c]{success\\ rate}  &\makecell[c]{POVM\\ operators} & \multicolumn{3}{c|}{coefficients in (\ref{def:T32})}  &\makecell[c]{rotation angle\\ in $e^{i\a\s_{n_C}}$}\\
		\hline
		$\t_1$&$\t_2$&$\varphi_1$&$\varphi_2$&$\l_1$&$\l_2$&$\og_1$&$\og_2$&$p$& \makecell[c]{$M_1,N_1$ \\or $M_2,N_2$}&$K$&$c_{0,1}$&$c_{1,0}$&$\a$  \\
		\hline
		\multirow{4}{*}{$\frac{\pi}{4}$}   & \multirow{4}{*}{$\frac{\pi}{4}$ }& \multirow{4}{*}{0} &\multirow{4}{*}{ $\pi$} &\multirow{4}{*}{$\frac{\pi}{4}$}   &\multirow{4}{*}{$\frac{\pi}{4}$} &\multirow{4}{*}{ $\frac{\pi}{2}$}&  \multirow{4}{*}{ $\frac{3\pi}{2}$}&\multirow{4}{*}{50\%}&$M_1,N_1$&1&$-\frac{1}{2}i$&$\frac{1}{2}$ &   $\frac{3\pi}{4}$ or $\frac{7\pi}{4}$   \\
		\cline{10-14}
		&&   & & &  &  &  &  &$M_1,N_2$&1 &$\frac{1}{2}i$&$\frac{1}{2}$&   $\frac{\pi}{4}$ or $\frac{5\pi}{4}$     \\
		\cline{10-14}
		&&   & & &  &  &  &  &$M_2,N_1$&-1 &$-\frac{1}{2}i$&$-\frac{1}{2}$&   $\frac{\pi}{4}$ or $\frac{5\pi}{4}$     \\
		\cline{10-14}
		&&   & & &  &  &  &  &$M_2,N_2$&-1 &$\frac{1}{2}i$&$-\frac{1}{2}$&   $\frac{3\pi}{4}$ or $\frac{7\pi}{4}$     \\
		\hline
		\multirow{4}{*}{$\frac{\pi}{4}$}   & \multirow{4}{*}{$\frac{\pi}{4}$ }& \multirow{4}{*}{0} &\multirow{4}{*}{ $\pi$} &\multirow{4}{*}{\makecell[c]{$(0,\frac{\pi}{4})\cup$\\$(\frac{\pi}{4},\frac{\pi}{2})$}}   &\multirow{4}{*}{$\frac{\pi}{2}-\l_1$} &\multirow{4}{*}{ $\frac{\pi}{2}$}&  \multirow{4}{*}{ $\frac{3\pi}{2}$}&\multirow{4}{*}{25\%}&$M_1,N_1$&1&$-\frac{\sqrt{2}}{2}i\sin\l_1$&$\frac{\sqrt{2}}{2}\cos\l_1$ &  $\pi-\l_1$ or $2\pi-\l_1$   \\
		\cline{10-14}
		& &   & & &  &  &  &  & $M_1,N_2$&1&$\frac{\sqrt{2}}{2}i\sin\l_2$&$\frac{\sqrt{2}}{2}\cos\l_2$ &  $\frac{3\pi}{2}-\l_1$ or $\frac{\pi}{2}-\l_1$ \\
		\cline{10-14}
		& &   & & &  &  &  &  & $M_2,N_1$&-1&$-\frac{\sqrt{2}}{2}i\sin\l_1$&$-\frac{\sqrt{2}}{2}\cos\l_1$ &  $\l_1$ or $\l_1+\pi$   \\
		\cline{10-14}
		& &   & & &  &  &  &  & $M_2,N_2$&-1&$\frac{\sqrt{2}}{2}i\sin\l_2$&$-\frac{\sqrt{2}}{2}\cos\l_2$ & $\l_1+\frac{3\pi}{2}$ or $\l_1+\frac{\pi}{2}$   \\
		\hline
	\end{tabular}
\end{table}

The operation $e^{i\a\s_{n_C}}$ with $\a\in(0,\frac{\pi}{2}), (\frac{\pi}{2},\frac{\pi}{2}), (\pi,\frac{3\pi}{2}), (\frac{3\pi}{2},2\pi)$ can be realized only  when the operators $(M_2,N_1), (M_1,N_1), (M_1,N_2), (M_2,N_2)$ occur in the measurement, respectively. From (\ref{def:probability}), each of the combination of POVM operators occurs with the probability of 25\%. So the success rate is equal to 25\%.

	To conclude, we have shown that how  Bob and Charlie implement the remote operation $e^{i\a\s_{n_C}}$ in the POVM measurement scheme without the permission of controller Alice. They can realize the operation of
	$\a\in\{0,\frac{\pi}{4},\frac{\pi}{2},\frac{3\pi}{4},\pi,\frac{5\pi}{4},\frac{3\pi}{2},\frac{7\pi}{4}\}$ with the  success rate of 50\%, and the operation corresponding to $\a\in\bigcup_{m=0}^7(\frac{m\pi}{4},\frac{(m+1)\pi}{4})$ can be realized with the success rate of 25\%.
\end{proof}

\section*{DATA AVAILABILITY}
All relevant data supporting the main conclusions and figures of the document are
available upon reasonable request. Please refer to Xinyu Qiu at xinyuqiu@buaa.edu.cn.

\bibliographystyle{unsrt}
\bibliography{graph_state}
\end{document}